\documentclass[12pt, draftclsnofoot, onecolumn]{IEEEtran}

\usepackage{cite}
\usepackage{amsmath,amssymb,amsfonts}
\usepackage{algorithmic}
\usepackage{graphicx}
\usepackage{textcomp}
\usepackage{xcolor}
\usepackage{epstopdf}
\usepackage{mathrsfs}
\usepackage{pifont}
\usepackage{multirow}
\usepackage{algorithm}
\usepackage{algorithmic}
\usepackage{amsthm}
\usepackage{bbm}
\newtheorem{lemma}{\textbf{Lemma}}
\usepackage{bm}
\usepackage{subfigure}

\newtheorem{thm}{\textbf{Theorem}}
\usepackage{color}

\begin{document}
	
	\title{Energy-Efficient URLLC Service Provision via a Near-Space Information Network}
	\author{Puguang An, Peng Yang,~\IEEEmembership{Member,~IEEE}, Xianbin Cao,~\IEEEmembership{Senior Member,~IEEE}, 
	Kun Guo,~\IEEEmembership{Member,~IEEE}, 
	Yue Gao,~\IEEEmembership{Senior Member,~IEEE}, 
 
	and Tony Q. S. Quek,~\IEEEmembership{Fellow,~IEEE}

		\thanks{
		P. An, P. Yang, and X. Cao are with the School of Electronic and Information Engineering, Beihang University, Beijing 100191, China.
K. Guo is with the School of Communication and Electronic Engineering, East China Normal University, Shanghai 200062, China. 
Y. Gao is with the School of Computer Science, Fudan University, Shanghai 200433, China.
T. Q. S. Quek is with the Information Systems Technology and Design, Singapore University of Technology and Design, 487372 Singapore.

		}
	}
	
	
	
	
	
	\maketitle
	
	\begin{abstract}
		The integration of a near-space information network (NSIN) with the reconfigurable intelligent surface (RIS) is envisioned to significantly enhance the communication performance of future wireless communication systems by proactively altering wireless channels. This paper investigates the problem of deploying a RIS-integrated NSIN to provide energy-efficient, ultra-reliable and low-latency communications (URLLC) services. We mathematically formulate this problem as a resource optimization problem, aiming to maximize the effective throughput and minimize the system power consumption, subject to URLLC and physical resource constraints. The formulated problem is challenging in terms of accurate channel estimation, RIS phase alignment, and effective solution design. We propose a joint resource allocation algorithm to handle these challenges. In this algorithm, we develop an accurate channel estimation approach by exploring message passing and optimize phase shifts of RIS reflecting elements to further increase the channel gain. Besides, we derive an analysis-friendly expression of decoding error probability and decompose the problem into two-layered optimization problems by analyzing the monotonicity, which makes the formulated problem analytically tractable. Extensive simulations have been conducted to verify the performance of the proposed algorithm. Simulation results show that the proposed algorithm can achieve outstanding channel estimation performance and is more energy-efficient than diverse benchmark algorithms. 
	\end{abstract}
	

	
	
	\maketitle
	
	\section{Introduction}
	%
	%
	%
	%
	%
	%
	%
	%
	%
	%
	%
	
	\IEEEPARstart{T}{he} upcoming sixth generation (6G) communication systems have the unprecedented requirements for ultra-reliability, ultra-low latency, and extremely high data rates that the fifth generation (5G) communication systems cannot well fulfill \cite{DBLP:journals/twc/BrighenteMBMT22}. 
	Ultra-reliable and low-latency communications (URLLC), which is one of the three pillar services for the 5G systems, will still act as a pillar service for the 6G systems. URLLC can empower diverse mission-critical applications including intelligent transportation systems, autonomous driving, telemedicine, industrial automation, Tactile Internet, real-time control in Digital Twins, and Metaverse. 
	
	It's widely considered that space-air-ground integrated communication systems are a crucial system architecture that can satisfy the unprecedented requirements of 6G systems \cite{DBLP:journals/comsur/ZhouSLH23}. A near-space information network (NSIN), composed of various near-space platforms (e.g., airships), medium/low-altitude platforms (e.g., medium/low-altitude unmanned aerial vehicles (UAVs)), is a network system that can acquire, transmit, and process space electromagnetic (EM) signals in real-time. It's the core component of the space-air-ground integrated communication systems. 
	Compared to space-based networks, NSIN is much closer to terrestrial users and can provide high-capacity communication services. Compared with terrestrial networks, NSIN can quickly extend the communication distance and is robust to unexpected events (e.g., diverse natural disasters). Thus, NSIN can provide wide-area and reliable communication coverage. In this regard, NSIN integrates the advantages of space-based networks and terrestrial networks and becomes an important supplement to terrestrial networks and space-based networks. It's envisioned that resorting to NSIN will become a new regime for enabling URLLC \cite{DBLP:journals/comsur/KurtKAIDAYY21,DBLP:journals/tii/WangAS20}. 	
	The authors of this manuscript have also demonstrated the feasibility of delivering reliable and low-latency control-and-non-payload EM signals using one of our manufactured near-space airships during the execution of many projects on airship-supported intelligent transportation systems and meteorological observations over the past few years \cite{Ministry2019Comprehensive}. 
	
	{Except for the mentioned space-air-ground integrated systems,} reconfigurable intelligent surface (RIS) has emerged as a promising paradigm to achieve smart wireless reflection channels for 6G systems. 
	RIS is a planar surface consisting of many passive reflecting elements for EM waves, each of which can separately change the amplitude and/or phase of impinging EM signals. 
	With the deployment of RIS, wireless channels between transmitters and receivers can be flexibly reconfigured to achieve desired realizations and distributions. Thus, RIS provides a new paradigm to fundamentally address the complex EM signal interference and channel fading issues and potentially achieves a significant improvement in terms of throughput and reliability of wireless communication systems \cite{DBLP:journals/twc/PerovicTRF21}. 
	
	Different from existing advanced communication techniques (e.g., beamforming, dynamic rate control) to realize ultra-reliable and high-capacity transmission by adapting to channel conditions, both NSIN and RIS can proactively alter the wireless channels to significantly improve communication performance. Specifically, NSIN can achieve rich line-of-sight (LoS) transmission through trajectory/location control. RIS can create desired LoS links for EM wave propagation and eliminate EM signal interference via smart EM signal reflection. 
	The integration of NSIN with RIS achieves complementary advantages, which will significantly increase the design flexibility of NSIN, reduce cost, and {improve transmission} performance. 
	Through mounting RIS on an airborne platform of NSIN, RIS can achieve a 360° panoramic full-angle reflection toward target nodes. 
	By deploying RIS-integrated NSIN, wide-area interference inherent in NSIN can be {mitigated}, and the spectral efficiency of NSIN can be improved. 
	Meanwhile, the integration of RIS in NSIN will improve the energy efficiency of NSIN due to the low hardware and energy costs of RIS. 
	
	\subsection{State of the Art}
	Recently, research on the URLLC-enabled UAV network or NSIN has attracted much attention, as outlined in \cite{DBLP:journals/tii/WangAS20,DBLP:journals/tii/RanjhaKD22,DBLP:journals/tvt/BartoliM22,DBLP:journals/tcom/WangPRXZN21,DBLP:journals/jsac/YangXGQCC21,DBLP:journals/iotj/RanjhaK21, DBLP:journals/iotj/ChenWZWF21,DBLP:journals/tcom/XiCYCQW21,DBLP:journals/wcl/YangXQCC21}.
	For example, the work in \cite{DBLP:journals/tii/WangAS20} proposed an NSIN with low/medium/high altitude platforms to support URLLC applications. The authors {developed} a complex network optimization tool, i.e., polychromatic sets, to analyze the performance of each link. Based on the analysis results, they designed distributed and centralized link selection schemes to select end-to-end links that could {meet} the stringent latency and reliability requirements of URLLC. The work in \cite{DBLP:journals/tii/RanjhaKD22} {utilized} a UAV relay to {send} short URLLC control packets from a controller to many mobile robots. In particular, the authors {investigated the optimization problem in terms of UAV deployment,} beamwidth, transmit power, and blocklength allocated to each robot and proposed an iterative {search} method to obtain the optimal solution. 
	The work in \cite{DBLP:journals/tvt/BartoliM22} analyzed channel quality indicator (CQI) aging effects in URLLC under the finite blocklength regime and utilized a Recurrent Neural Network (RNN) approach to predict the next CQI for UAV control information exchange.
	Besides, the authors in \cite{DBLP:journals/tcom/WangPRXZN21} studied the {failed packet reception} probability and effective throughput of URLLC control links in UAV communications. They derived the closed-form expression of packet error probability and throughput using Gaussian-Chebyshev quadrature and obtained the optimal packet length using an exhaustive search. 
	
    RIS is an important emerging technique that can provide optimized energy and spectral efficiency for communication systems. {Recent studies showed that a RIS-assisted communication system could be 40\% more energy-efficient than a relay-assisted system \cite{DBLP:conf/globecom/HuangAZDY18}.}
	Above works \cite{DBLP:journals/tii/WangAS20,DBLP:journals/tii/RanjhaKD22,DBLP:journals/tvt/BartoliM22,DBLP:journals/tcom/WangPRXZN21,DBLP:journals/jsac/YangXGQCC21,DBLP:journals/iotj/RanjhaK21, DBLP:journals/iotj/ChenWZWF21} did not investigate the NSIN- or UAV-assisted URLLC systems with the aid of RIS. 
    {Nonetheless, the integration of RIS in aerial URLLC communication systems has recently attracted extensive attention in the research community \cite{DBLP:journals/access/LiYDMD21, DBLP:journals/iotj/RanjhaK21a, DBLP:journals/tvt/SamirEASG21}.}
	For instance, the authors in \cite{DBLP:journals/access/LiYDMD21} proposed a UAV-RIS system to support the stringent URLLC constraints and formulated an optimization framework in terms of UAVs' locations, the phase shift of RIS, and the blocklength of URLLC. To solve the formulated problem, a deep neural network (DNN) architecture was designed to obtain 3-D deployment locations of UAVs, and an iterative optimization scheme was explored to optimize the phase shift and blocklength. 
	The work in \cite{DBLP:journals/iotj/RanjhaK21a} considered a scenario of deploying a UAV and a ground RIS to deliver URLLC packets between ground Internet-of-Things (IoT) devices. The authors proposed to jointly optimize the passive beamforming, the deployment location of a UAV, and the blocklength to minimize the packet decoding error probability (DEP). A computationally efficient Nelder–Mead simplex approach was utilized to solve the joint optimization problem. 
	Additionally, the authors in \cite{DBLP:journals/tvt/SamirEASG21} investigated the benefits of UAV/RIS-based systems to passively relay low-latency information sampled by Internet of Things (IoT) devices to a base station (BS). 
	For this aim, they formulated a joint optimization problem of UAV altitude, the phase shift of RIS, and communication scheduling, aiming to minimize the expected sum Age-of-Information (AoI) under a reliable decoding constraint. They also designed a DRL framework based on proximal policy optimization to solve the optimization problem. 
	
	\subsection{Motivations and Contributions}
	Although the issue of designing UAV/RIS-based systems to support URLLC applications was investigated in the works \cite{ DBLP:journals/access/LiYDMD21, DBLP:journals/iotj/RanjhaK21a, DBLP:journals/tvt/SamirEASG21}, they did not consider the scenario of integrating high altitude platforms (HAPs), UAVs, and aerial RIS to provide wide-area and robust communication coverage in an energy- and cost-efficient manner. 
	Most recently, the works in \cite{DBLP:journals/corr/abs-2104-01723, DBLP:journals/wcl/GaoJLM21} investigated the integration of NSIN and RIS. 
	For example, the authors in \cite{DBLP:journals/corr/abs-2104-01723} considered a novel wireless architecture composed of a RIS-enabled HAP and multiple UAVs. The HAP-RIS could overcome the limitations of conventional terrestrial RIS and achieve rich LoS and full-area coverage. The UAVs that would receive backhaul reflection signals from HAP-RIS were deployed to serve terrestrial flash crowd traffic. They formulated and solved a joint optimization problem of HAP deployment location and phase shift of RIS to increase the energy efficiency of UAVs. 
	The work in \cite{DBLP:journals/wcl/GaoJLM21} considered a scenario of deploying a UAV-RIS to assist the HAP downlink transmission. The authors formulated an optimization problem of maximizing the signal-to-noise ratio (SNR) of ground users (GUs) by jointly optimizing the UAV trajectory and the phase shift of UAV-RIS. They developed a model-free reinforcement learning framework to learn the UAV trajectory and adjust the phase shift according to the learned UAV trajectory. 
	The works \cite{DBLP:journals/corr/abs-2104-01723, DBLP:journals/wcl/GaoJLM21}, however, investigated the high throughput scenario instead of the URLLC scenario.
	Moreover, they did not studied the crucial HAP-UAV channel estimation problem. A LoS propagation assumption between a HAP and a low-altitude UAV was enforced in \cite{DBLP:journals/corr/abs-2104-01723, DBLP:journals/wcl/GaoJLM21}. This assumption will be unreasonable as the received signals by a low-altitude UAV include both LoS and non-line-of-sight (NLoS) (e.g., reflected/scattered signals) components. Therefore, it's essential to study the HAP-UAV channel estimation issue. 
	To tackle the above issues, this paper proposes to design a RIS-integrated NSIN to support energy-efficient URLLC services, which leverages the advantages of power control and RIS. 
	The main contributions of this paper are summarized as follows:
	\begin{itemize}
		\item We formulate the problem of providing energy-efficient URLLC services using a RIS-integrated NSIN as a resource optimization problem. The goal of this problem is to maximize the effective throughput and minimize the system power consumption, subject to URLLC and physical resource constraints. 
	The formulated problem is challenging in terms of channel estimation, RIS phase shift optimization, theoretical analysis, and effective solution. 
	\item To handle these challenges, we propose a joint resource allocation algorithm in which we first develop a probabilistic channel model to capture the sparsity of the HAP-UAV channel and then design a novel channel estimation approach by message passing to accurately estimate the HAP-UAV channel. 
	\item Second, we apply a maximum-ratio transmission (MRT) strategy to maximize the signal-to-noise ratio (SNR) and design a RIS phase shift optimization strategy, resulting in a suboptimal HAP-UAV channel gain.
	Besides, we derive an analysis-friendly expression of DEP by Taylor expansion to make the formulated problem analytically tractable {and help to lower the DEP value}. 
	\item Third, through analyzing the monotonicity of the objective function in the feasible region, we decouple the decision variables and horizontally decompose the problem into a BS-Layer optimization problem and a UAV-Layer optimization problem. An iterative optimization strategy is then explored to optimize transmit power and blocklength alternately. 
	\item Finally, the proposed algorithm is compared with diverse benchmarks to verify its performance, and the impact of various design parameters is comprehensively discussed {through extensive simulations}. Simulation results show that the proposed algorithm can achieve accurate HAP-UAV channel estimation and can provide energy-efficient URLLC services. 
	\end{itemize}

\section{System model and problem formulation}
	\subsection{System Model}
	We consider a downlink transmission scenario based on a RIS-integrated NSIN communication system for energy-efficient URLLC service provision. As illustrated in Fig. \ref{fig_downlink_scenario}, the system mainly includes a base station (BS), a RIS-mounted HAP, a low-attitude UAV, and multiple ground \textcolor{black}{remote URLLC robots}. The BS is designated to transmit ultra-reliable and low-latency signals towards these \textcolor{black}{remote robots for accomplishing some critical tasks (e.g., grid, pipeline, security, and remote railway patrol, farming, and mining)}. 
	\begin{figure*}[!t]
		\centering
		\includegraphics[width=5.6in, height=2.2in]{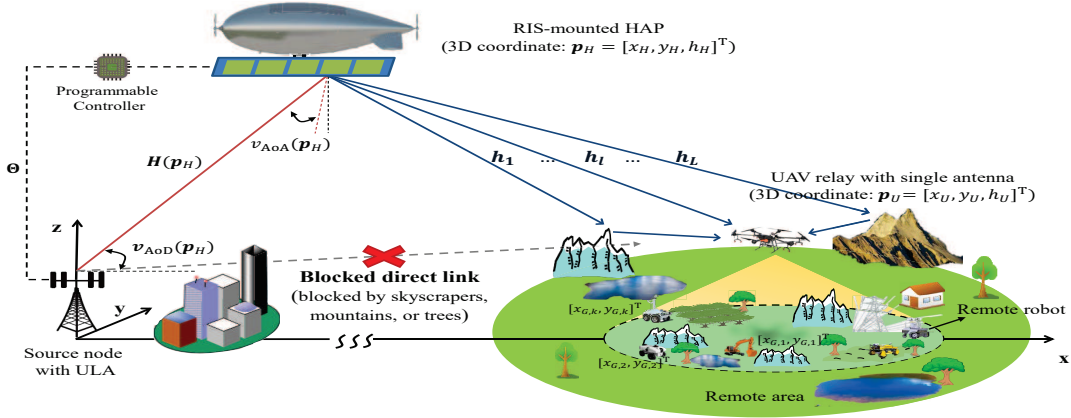}
		\caption{{A downlink URLLC transmission scenario with a RIS-integrated NSIN communication system.}}
		\label{fig_downlink_scenario}
	\end{figure*}
	\textcolor{black}{A direct propagation path between the BS and a remote URLLC robot, however, will be blocked by numerous physical obstructions on the ground (e.g., skyscrapers, mountains, and trees) over the remote propagation path.} 
	Then, the RIS-mounted HAP and the UAV can work cooperatively to establish ultra-reliable near-space communication links between the BS and multiple \textcolor{black}{remote URLLC} robots. The RIS-mounted HAP acts as a passive flying relay, and the UAV acts as an active flying relay.
	In this system, the BS is located at the origin and equipped with a uniform linear-array (ULA) antenna with $M$ antenna array elements. The antenna gain equals $G$, and the spacing of adjacent antenna elements is $d_{\rm BS}$. There are $K$ single-antenna URLLC robots that are randomly and uniformly distributed in a considered area $\mathcal{G}$. The URLLC robot set is denoted as ${\mathcal K}= \left\{ 1,...,K\right\} $, {and denote the coordinate of the $k$-th robot by ${\bm p}_G = [x_{G,k},y_{G,k}]^{\rm T}$}. The UAV is equipped with an omnidirectional antenna, hovering in the air with a fixed three-dimensional (3-D) coordinate ${\bm p}_U=[x_U,y_U,h_U]^{\rm T}$. A decoding-and-forward (DF) relay strategy is adopted to flexibly control the blocklength for the BS-HAP-UAV link and the UAV-to-ground (UtG) link. Moreover, the frequency division multiple access (FDMA) manner is exploited to eliminate intra-cell interference \cite{DBLP:journals/jsac/LiJSZ16,DBLP:journals/twc/ZhouSLLH19}.
	
	For the RIS-mounted HAP, we assume that a ULA-structured RIS is deployed directly below the HAP to provide wide-area and full-angle coverage. The RIS includes $N$ passive reflecting elements, separated by the spacing $d_{\rm RIS}$. 
	We assume that the BS array is deployed on the origin. 
	The HAP hovers at a fixed point ${\bm p}_H=[x_H,y_H,h_H]^{\rm T}$ in the stratosphere. We assume that all elements of the RIS are turned on without power amplification. Thus, ${\rm \bm \Theta}={\rm diag}\left\{e^{j\theta_1},e^{j\theta_2},...,e^{j\theta_N} \right\}\in \mathbb{C}^{N \times N}$ is the diagonal phase shift matrix of the RIS with phase shift parameter of the $n$-th element ${\theta_n}\in (0,2\pi]$ \cite{DBLP:journals/twc/KhaliliMZJYJ22,DBLP:journals/twc/HuZZ22,DBLP:journals/twc/LiDRTY21}.
	During the first transmission phase, the BS will transmit URLLC signals to the UAV with the assistance of the HAP relay. In this phase, a cascading channel comprised of the BS-HAP channel and the HAP-UAV channel is constructed. Denote the BS-HAP channel by ${\bm {H}}\in \mathbb{C}^{N \times M} $, and the HAP-UAV channel by ${\bm {h}}\in \mathbb{C}^{N}$. Therefore, the received signal at the UAV from the BS can be expressed as
	\begin{equation}\label{eq_received_signal}
		{y} = \sqrt{P_bG}{\bm {h}}^{\rm H}{\bm {\Theta}}{\bm {Hv}}u_1+{ n}
	\end{equation}
	where $P_b$ is the source transmit power of BS, ${\rm \pmb v}\in \mathbb{C}^{N \times 1}$ is the unit-magnitude precoding vector,  $u_1$ is the transmission signal with unit-power. ${ n}$ is the additive white Gaussian noise (AWGN) vector, following a distribution of $\mathcal{CN}(0,\sigma_0^2)$, and {$\sigma_0^2$} is the noise power. {We define ${{SNR}_u = P_bG|\bm{ h}^{\rm H} \bm{\Theta} \bm{Hv}|^2/\sigma_0^2}$ as the instantaneous signal-to-noise ratio (SNR) experienced by the UAV.}
	In the second transmission phase, the UAV decodes received URLLC packets from the BS and re-encodes and forwards the packets to robots via the FDMA scheme. Denote the UtG channel by $\left\{g_k\right\}$, $\forall k \in {\mathcal K}$. Then, the received signal of any robot can be expressed as
	\begin{equation}\label{eq_}
		{y_k} = \sqrt{P_u}g_k{u_2}+{ n_k}
	\end{equation}
	where $P_u$ is the transmit power of UAV, $u_2$ is the transmission signal with unit-power, the noise ${n_k}$ is subject to a normal distribution with zero mean and variance $\sigma_{k}^2$. 
	Then, ${SNR}_k = P_u|g_k|^2/\sigma_{k}^2$ represents the instantaneous SNR experienced by robot $k$.

	\subsection{BS-HAP Channel Model}
	In practice, HAP is usually deployed in an EM propagation environment without rich scatters. It's then reasonable to assume that the BS-HAP channel is dominated by the LOS component. 
 The resulting BS-HAP channel matrix can be expressed as
	\begin{equation}\label{eq_H_matrix}
		{\bm H}=\alpha\bm{a}_{\rm RIS}(\upsilon_{{\rm AoD}}(\bm{p}_{H}))\bm{a}_{\rm BS}^{\rm H}(\upsilon_{{\rm AoA}}(\bm{p}_{H}))
	\end{equation}
	where $\alpha$ represents the complex channel gain, capturing path loss and additional attenuation caused by the EM propagation environment (e.g., rain and water vapor). {$\bm{a}_{\rm RIS}(.)\in \mathbb{C}^{M}$ and $\bm{a}_{\rm BS}(.)\in \mathbb{C}^{N}$ are array steering vectors of BS and RIS-mounted HAP, respectively, with $\bm{a}_{\rm RIS}(.)=[1,e^{-j2\pi \bar d_{\rm RIS}{{\rm sin}(.)}},...,e^{-j2\pi (M-1)\bar d_{\rm RIS}{{\rm sin}(.)}}]^{\rm T}$ and $\bm{a}_{\rm BS}(.)=[1,e^{-j2\pi \bar d_{\rm BS}{{\rm sin}(.)}},...,e^{-j2\pi (N-1)\bar d_{\rm BS}{{\rm sin}(.)}}]^{\rm T}$}, 
 where $\bar d_{\rm RIS} = d_{\rm RIS}/{{\lambda_0}}$, $\bar d_{\rm BS} = d_{\rm BS}/{{\lambda_0}}$, ${\lambda_0}$ is wavelength. $\upsilon_{{\rm AoD}}(\bm{p}_{H})$ and $\upsilon_{{\rm AoA}}(\bm{p}_{H})$ are the AoD and AoA of the BS-HAP link, respectively. Aided by precise position information obtained by GPS or other positioning systems, both AoA and AoD can be obtained.
	
	Owing to the relative stability of the BS-HAP channel, we assume that its channel gain can be maintained stable for a long time and continuous  estimation is not necessary. Therefore, the channel gain is considered to be known in advance in this paper.
	
	\subsection{HAP-UAV Channel Model}
	Unlike the BS-HAP channel, the EM propagation environment of the HAP-UAV channel will be more complex. A UAV is usually deployed close to robots to establish high-quality UtG links; as a result, there is limited physical scatters (e.g., trees and mountains) around the UAV, {as shown in Fig. \ref{fig_downlink_scenario}}. 
	We consider such a case that the HAP-UAV channel fading consists of large-scale channel fading and small-scale channel fading caused by limited reflected/scattered paths. The large-scale fading can be obtained by computing free-space path loss and meteorological loss. {Nevertheless, the acquisition of the small-scale fading under the limited scattering environment is highly challenging; as a result, we need to estimate it.} In addition, due to the high-frequency transmission characteristic \cite{grace2011broadband}, the transmitted signal has a sound geometric characteristic. To this end, we adopt a narrowband flat-fading channel model in \cite{DBLP:books/cu/TV2005}
	\begin{equation}\label{eq_h_multipath}
		\displaystyle \bm{ h}=\sum\nolimits_{l=1}^{L}\beta_l\bm{ a}_{\rm RIS}(\omega_{l})
	\end{equation}
	where $\beta_l$  and $\omega_{l}$ represent the complex path gain and AoD of the $l$-th path, respectively. 

	\subsection{Decoding Error Probability in a URLLC Regime}
	Distinct from conventional wireless communications, a URLLC regime has stringent QoS requirements on transmission latency and reliability (e.g., DEP). To satisfy these requirements, the blocklength of URLLC packets should be controlled within a finite length. However, {the Shannon's capacity formula cannot relate these performance metrics with the blocklength}. Shannon's capacity refers to the maximum achievable channel coding rate (MAR) that can be obtained with arbitrary transmission latency, which indicates that a signal blocklength should be sufficiently long. To tackle this issue, we leverage tight MAR approximations \cite{DBLP:journals/tit/PolyanskiyPV10,yang2014quasi} of the BS-HAP-UAV link (denoted by $R_u$) and the UAV-GU link (denoted by $R_k$), which are extremely precise under the assumption that a blocklength is greater than 50 symbols, i.e., $R_* = {\rm log_2}(1+{SNR}_*)-\sqrt{\frac{1-(1+{SNR}_*)^{-2}}{b_*}}\frac{Q^{-1}(\varepsilon_{*})}{\rm ln2}$,
	where $*\in\left\{u,k\right\}$, $Q(x)=\frac{1}{\sqrt{2\pi}}\int_x^{\infty}{\rm exp}(-t^2/2){\rm d}t$ is a Gaussian Q-function, which is also a monotonically decreasing function of the argument. $b_u$ and $b_k$ are the blocklengths allocated to the BS-HAP-UAV link and the UAV-GU link, respectively. As the UAV adopts the DF strategy, the BS-HAP-UAV link and the UAV-GU link can be configured with different blocklengths. $\varepsilon_u$ and $\left\{\varepsilon_k \right\}_{k=1}^K$ are the DEPs of the UAV and the robots, respectively. $F_k$ {is denoted as} the size of a URLLC packet transmitted to the $k$-th robot from the UAV, and $F_B$ the size of a URLLC packet sent by the BS. According to the FDMA downlink transmission scheme, we have $R_u=F_B/b_{u}$ and $R_k=F_k/b_{k}$ \cite{DBLP:journals/tvt/YuanYFD22,DBLP:journals/icl/PanRDEN19,DBLP:journals/corr/abs-2207-04738,DBLP:journals/icl/HashemiAML22}. The corresponding DEP is then given by
	\begin{equation}\label{eq_}
		\begin{array}{l}
			{\varepsilon _*} = Q\left( {\sqrt {\frac{{{b_*}}}{{1 - {{(1 + {{SN}}{{{R}}_*})}^{ - 2}}}}} \left( {{\rm{lo}}{{\rm{g}}_2}\left( {1 + {{SN}}{{{R}}_*}} \right) - {R_*}} \right){\rm{ln}}2} \right)
		\buildrel \Delta \over = Q(f({b_*},{{SN}}{{{R}}_*}))
		\end{array}
	\end{equation}
	
	\subsection{Problem Formulation}
	Based on the above system model, this subsection aims to formulate a joint transmission power, blocklength, and RIS phase shift optimization problem for providing energy-efficient URLLC services for robots.
	
	\subsubsection{URLLC and physical resource constraints}
	DEP and communication latency are crucial performance metrics of URLLC services, which should be {carefully delt with}. 
	{Let $\varepsilon_u$ represent the DEP at the UAV, and $\varepsilon_k$ the DEP at $k$-th robot.} 
	As decoding errors, either at the UAV or at a GU, will result in transmission failure, the overall DEP experienced by the $k$-th robot is
	\begin{equation}\label{eq_overall_epsilon}
		\displaystyle  \varepsilon = (1-\varepsilon_{u})\varepsilon_k+\varepsilon_u = \varepsilon_k+\varepsilon_u-\varepsilon_{u}\varepsilon_k < \varepsilon_k+\varepsilon_u
	\end{equation}
	
	Note that we omit the term $\varepsilon_u\varepsilon_k$ as it's much smaller than the other two terms. {From (\ref{eq_overall_epsilon}), we can observe that both the DEP of a UAV and the DEP of a robot should be reduced to decrease the overall DEP.} 
 {Nevertheless, statistical features of the cascaded BS-HAP-UAV channel and the UtG channel are distinct. The distribution of the cascaded BS-HAP-UAV channel is unknown, resulting in the unavailability of its statistics.} Further, as a backbone link, it's significant to provide an extremely low and stable DEP over the BS-HAP-UAV link under any circumstance. 
    The UtG small-scale channel fading model is {considered to be known \cite{DBLP:journals/comsur/KhuwajaCZAD18}}, and then the constraint on the DEP of a robot can be replaced with the form of expectation.
	Hence, we separately impose the DEP constraints on the UAV and robots as follows 
	\begin{equation}\label{eq_}
		\varepsilon_u\leq \varepsilon^{\rm th}_u, \text{ } {\mathbb E}[\varepsilon_k] \leq \varepsilon^{\rm th}_k,\forall k\in {\mathcal K}
	\end{equation}
	where $\varepsilon^{\rm th}_u$ and $\varepsilon^{\rm th}_k$ are the tolerable DEPs of the UAV and the $k$-th GU, respectively, ${\mathbb E}[\cdot]$ the expectation operator. 
	
	Next, we discuss the communication latency-related constraint. To guarantee the ultra-low latency performance, the blocklength should be maintained at a finite length \cite{DBLP:journals/access/LiYDMD21,DBLP:journals/icl/PanRDEN19,DBLP:journals/tvt/YuanYFD22}. Note that the UAV adopts the FDMA-DF mode, which allows us to flexibly design different blocklengths over each communication link to further boost the overall NSIN system performance. Define the minimum and maximum blocklength sets allocated to the BS and the UAV as $\left\{B_{u}^{\min},B_{u}^{\max}\right\}$ and $\left\{B^{\min},B^{\max}\right\}$, respectively. Thus, the latency-related constraints can be described as \cite{DBLP:journals/icl/PanRDEN19,DBLP:journals/tvt/YuanYFD22}
	\begin{equation}\label{eq_BS_length}
		B_{u}^{\min} \leq b_u \leq B_{u}^{\max}, \text{ } b_u \in \mathbb{Z^+}, \text{ } B^{\min}\leq b_k \leq B^{\max}, \text{ } b_k  \in \mathbb{Z^+},\forall k\in {\mathcal K}
	\end{equation}
	Finally, we discuss the constraints on BS and UAV transmit power. We investigate the power control problem of BS and UAV to achieve energy-efficient URLLC packet transmission. 
	The configured transmit power is not allowed to exceed a power budget. Denote the maximum transmit power of the BS and the UAV by $P_{B}$ and $P_{U}$, respectively. Therefore, we have the following transmit power constraints
	\begin{equation}\label{eq_BS_power}
		0 < P_b\leq P_{B}, \text{ }  0 < P_u \leq P_{U}
	\end{equation}
	\subsubsection{Objective function design}
	Energy efficiency is a key performance evaluation indicator for UAV-assisted cellular networks \cite{DBLP:journals/tvt/ChangMLW21,DBLP:journals/iotj/WangLLTK20,DBLP:journals/tgcn/ChakareskiNM0AR19}. Based on the concepts of MAR, and BS and UAV transmit power, we design an energy efficiency model to represent the objective function. 
	The energy efficiency model is defined as the ratio of the effective throughput \cite{DBLP:journals/tvt/YuanYFD22} to the system power consumption, which is given by 
	\begin{equation}\label{eq_}
		\displaystyle  \eta_k = \frac{R_u(1-\varepsilon_u/\varepsilon^{\rm th}_u)+R_k(1-\varepsilon_k/\varepsilon^{\rm th}_k)}{P_b/P_B+{P_u/P_U}}, \forall k \in {\mathcal K}
	\end{equation}
	where the system power consumption is the sum of normalized BS and UAV transmit power.
 

	As different robots may experience diverse channel qualities, we aim at maximizing the minimum system energy efficiency when guaranteeing the stringent URLLC requirements of all robots. 
	We propose to jointly optimize UAV and BS transmit power, blocklength, and RIS phase shift to achieve the goal, under the constraints on DEPs, blocklengths, and UAV and BS transmit power.
	The joint resource optimization problem can then be formulated as 	\begin{subequations}\label{eq:formualted_problem}
		\begin{alignat}{2}
			& \max \limits_{P_u,P_b,b_u,\{b_k\},\{\phi_n\}} \min  \limits_{k \in \mathcal K}  \text{ } \displaystyle  \frac{R_u(1-\epsilon_u)+R_k(1-\epsilon_k)}{ {\bar p}_b+\bar p_u}    \\
			\allowdisplaybreaks[4]
			& {\rm s.t:} \text{ }
			 Q\left( {f\left( {{b_u},{SNR}_u} \right)} \right) \le \varepsilon _u^{{\rm{th}}} \\ 
			& {  {\mathbb E}[Q\left( {f\left( {{b_k},{SN}{{R}_k}} \right)} \right)] \le \varepsilon _k^{{\rm{th}}},\forall k \in {\mathcal K}  } \\
			& \phi_n \in [0,2\pi), \forall n\in \mathbb{N} \\
			& {\rm Constraints \text{ } (\ref{eq_BS_length}) \text{ } and \text{ } (\ref{eq_BS_power})}
		\end{alignat}
	\end{subequations}
	where ${\bar p}_b= P_b/P_B$, $\bar p_u = P_u/P_U$, $\epsilon_u = \varepsilon_u/\varepsilon^{\rm th}_u$, and $\epsilon_k = \varepsilon_k/\varepsilon^{\rm th}_k$.
	The solution of (\ref{eq:formualted_problem}) is highly challenging for the following reasons:
	1) The HAP-UAV channel gain is unknown and needs to be accurately estimated. However, there is a lack of studies on sparse HAP-UAV channel modeling and corresponding accurate channel estimation methods in the literature. 
	{2) There are multiple EM propagation paths in the HAP-UAV channel, and thus, it's difficult to obtain the optimal RIS phase shift such that the received signal strength of the UAV can be maximized.}
	3) The expression of DEP is highly complex, which makes the problem intractable and greatly hinders the theoretical analysis of the problem. 
	4) Given the estimated HAP-UAV channel gain, (\ref{eq:formualted_problem}) is still difficult to be solved. This is because it contains integer and continuous variables and the decision variables-related to HAP and UAV are intricately coupled, which indicate that the problem is a mixed integer non-convex programming (MINCP) problem. 

	In this paper, we develop a joint channel estimation, enhancement, and resource optimization framework to solve the above challenging problem. In particular, we first develop a probabilistic channel model to capture the sparsity of the HAP-UAV channel and design a novel channel estimation approach to estimate the HAP-UAV channel with high precision. 
	Second, we design a RIS phase shift optimization strategy that results in a suboptimal HAP-UAV channel gain.
	Third, we derive the approximate expressions of DEPs to make the problem analytically tractable. 
	Fourth, we analyze the monotonicity of the objective function of the problem. Based on the analysis results, we decompose the problem into two-layered optimization problems, which reduces the difficulty of the theoretical analysis of the problem. 
	
	\section{HAP-UAV CHANNEL ESTIMATION {and ENHANCEMENT}}
	In this section, we will introduce a message passing-based estimation approach of the HAP-UAV channel\footnote{{The channel estimation mentioned in the paper refers to the small-scale channel estimation without specification. As for large-scale fading, it is usually computed by some deterministic models in multipath channel estimation problems in the research community \cite{DBLP:journals/spl/ArdahGAH21,DBLP:journals/spl/WangFDL20,DBLP:journals/access/TahaAA21}.}}. {The RIS-user channel estimation is a hot and under-studied topic. Recent works in \cite{DBLP:journals/spl/ArdahGAH21,DBLP:journals/spl/WangFDL20,DBLP:journals/access/TahaAA21} proposed to use compressed sensing (CS) methods to estimate the RIS-user multi-path channel. However, the authors in \cite{DBLP:journals/spl/ArdahGAH21,DBLP:journals/spl/WangFDL20,DBLP:journals/access/TahaAA21} did not discuss the crucial issue of angle estimation or off-grid angle estimation. As a result, the accuracy of the channel estimation needs to be improved. Here, the HAP-UAV multi-path channel estimation with the off-grid angle consideration will be studied.} 

    \subsection{Off-Grid Angular Domain Channel Model}
    To utilize the sparse property of the HAP-UAV channel, we can transform the original channel model into an angular domain representation. In this way, the channel estimation problem is transformed into a sparse signal recovery problem. Generally, we adopt a uniform sampling grid ${\bm {\widehat\omega}}=\left\{\widehat\omega_1,...,\widehat\omega_N \right\}$, which is obtained by discretizing the angular domain $[-\frac{\pi}{2}, \frac{\pi}{2}]$, to match the true AoDs denoted as ${\bm {\omega}}=\left\{\omega_1,...,\omega_L \right\}$ ($L \ll N$). However, the continuous distribution characteristic of the true AoD indicates that the grid points cannot effectively capture the true AoDs in practice. Thus, the AoD estimation via an on-grid model usually leads to non-negligible errors. To perform high-precision estimation, we introduce off-grid offsets into the on-grid model, which is also called an off-grid model. Specifically, given $\omega_l \notin \left\{\widehat\omega_1,...,\widehat\omega_N \right\}$ and its nearest sampling grid point $\widehat{\omega}_{nl}= \mathop {\rm \min}  \limits_{n_l \in \left\{1,...,N\right\} }\left\{|\omega_l-\widehat\omega_{n_l}|\right\}$, we can define the off-grid offset $\Delta\phi_{nl}$ as
	\begin{equation}\label{eq_}
		\Delta\phi_{nl} = \begin{cases}{\omega}_{l} - \widehat{\omega}_{nl} , & l \in {\mathcal L} = \{1, 2, \ldots, L\}\\ 0,& {\rm otherwise}\end{cases}
	\end{equation}
	
	The case of $\Delta\phi_{nl} = 0$ indicates that only active paths will be assigned with a non-zero offset (If we don't know the number of true paths in advance, we can set a threshold of path strength to filter non-active paths). Then the corresponding off-grid steering matrix is ${{\bm {A}(\widehat{\bm{  \omega}}+\Delta\bm{\phi})}}=[\bm{ a}(\widehat\omega_1+\Delta\phi_{1}),...,\bm{ a}(\widehat\omega_N+\Delta\phi_{N})]$, and the sparse angular domain channel vector in the complex angular domain is $\bm{ x}\in \mathbb{C}^N$ with $L$ non-zero elements. Finally, the off-grid sparse angular domain representation of the {normalized} downlink channel $\displaystyle \bm{ \overline{h}}\in \mathbb{C}^N$ can be expressed as
	\begin{equation}\label{eq_offgrid_cha}
		\bm{\overline{h}}={\bm{A}(\widehat{\bm {\omega}}+\Delta \bm{\phi})\bm {x}} \triangleq {\bm{A}(\Delta \bm{\phi})\bm {x}}
	\end{equation}
	
	It's noteworthy that there must be some off-grid offsets making (\ref{eq_offgrid_cha}) hold. Though it's quite challenging to find the optimal off-grid offsets, the gap between the sampling grid points and the true AoDs will be significantly reduced by iteratively adjusting the off-grid offsets. Therefore, the off-grid model always outweighs the on-grid model.
	
	\subsection{Probability Model of Channel Vector}
	Recall that the HAP-UAV channel has limited paths in the angular domain. To fully exploit this {observation}, the probabilistic model of ${\bm{ x}}=[x_1,...,x_N]^{\rm T}$ can be modeled by two independent random hidden vectors, i.e., \emph{hidden support vector} and \emph{hidden value vector}. The \emph{hidden support vector} ${\bm{ d}}=[d_1,...,d_M]^{\rm T} \in \left\{0,1\right\}^N$ represents {the activeness of the paths}. Specifically, if $d_n=1$, then there is an active path around the $n$-th AoD direction $\widehat\omega_N+\Delta\delta_{N}$. We define $p(d_n=1)=\lambda$, indicating the channel sparsity. The complex-valued \emph{hidden value vector} ${\bm{q}} = [q_1,...,q_N]^{\rm T}$ represents the complex gains of paths, and $q_n \sim \mathcal{CN}(\zeta,\rho)$ follows an independent and identically distributed (i.i.d.) complex Gaussian distribution with a mean $\zeta$ and variance $\rho$.
	Accordingly, the joint probabilistic HAP-UAV channel model with a prior distribution can be expressed as $p({\bm{x}})=p({\bm{x}},{\bm{d}},{\bm{q}})=p({\bm{x}}|{\bm{d}},{\bm{q}})p({\bm{d}})p({\bm{q}})$,
	where $p({\bm{x}}|{\bm{d}},{\bm{q}})$ is the joint conditional prior. From the expression of $p({\bm{x}})$, we know that $p(x_n|d_n,q_n)=1$ conditioned on $x_n=d_nq_n$; otherwise, $p(x_n|d_n,q_n)=0$. Thus, $p({\bm{x}}|{\bm{d}},{\bm{q}})$ is given by
	\begin{equation}\label{eq_}
		\displaystyle  p({\bm{x}}|{\bm{d}},{\bm{q}})=\prod_{n=1}^Np(x_n|d_n,q_n)=\prod_{n=1}^N\delta(x_n-d_nq_n)
	\end{equation}
	where $\delta(.)$ is the Dirac delta function. 
	
	
	\subsection{Problem Formulation of Off-Grid Sparse Channel Estimation}
	Using the angular domain channel representation, we can rewrite (\ref{eq_received_signal}) as a standard CS model
	\begin{equation}\label{eq_input_output_equ}
	{{\bm{y}} = {\bm{F}}(\Delta \bm{\phi}){\bm{x}}+\bm{ n}_e}
    \end{equation}
    where {{${\bm{F}}(\Delta \bm{\phi})={\bm{U_0}^{\rm H}}{\bm {\bar H}}^{\rm H}{\bm{\Theta}}^{\rm H}{\bm{A}(\Delta \bm{ \phi})}\in\mathbb{C}^{P \times N}$}} is the measurement matrix, ${\bm{\bar H}}$ is the normalized BS-HAP channel matrix, {${\bm{U_0}}\in\mathbb{C}^{M \times P}$} is a training pilot matrix from the BS, $\bm{ n}_e \sim \mathcal{CN}(0, \sigma_e^2\bm I)$ is the complex AWGN. {In terms of} a CS problem, the choice of measurement matrix {has a direct influence on the desired} performance of the recovery approach. In the proposed model, we design pilot $\bm U_0$ based on the partial discrete Fourier transform (DFT) random permutation (pDFT-RP) measurement matrix, whose effectiveness is verified in \cite{DBLP:journals/spl/MaYP15a}. 
	{Please refer to Appendix A for the detailed design of the pilot matrix $\bm U_0$.}
	
    {We can then estimate the off-grid sparse channel with (\ref{eq_input_output_equ}).} Besides, we leverage the minimum mean square error (MMSE) criterion to evaluate the estimation performance. The estimated complex gain is $\hat{x}_n=E(x_n|{\bm{y}},\Delta \bm{ \phi})$, which will result in the MMSE. The expectation is taken over the marginal posterior $p(x_n|{\bm{y}},\Delta\bm{\phi})$ given measurements ${\bm{y}}$ and off-grid offsets $\Delta\bm{\phi}$. To exploit the inherent statistical structure of the HAP-UAV channel in terms of channel support and value vectors, the marginal posterior distribution can be expressed by Bayes' rule as
\begin{equation}\label{eq_marginal_posterior}
\begin{array}{l}
	\displaystyle  p(x_n|{\bm{y}},\Delta\bm{\phi}) \propto \sum_{\bm{d}} \int_{{\bm{x}}_{-n},{\bm{q}}}p({\bm{x,d,q}|{{\bm y},}\Delta\bm{\phi}}) 
	=\sum \limits_{\bm {d}} \int_{{\bm{x}}_{-n},{\bm{q}}} p({\bm{d}})p({\bm{q}})\prod_{n=1}^Np(x_n|d_n,q_n)\prod_{p=1}^P p(y_p|{\bm{x}},\Delta\bm{\phi})
\end{array}
\end{equation}
	where $\propto$ indicates proportionality up to a constant scale factor, ${\bm{x}}_{-m}$ means a vector $\bm{x}$ excluding the element $x_m$. According to (\ref{eq_input_output_equ}), the likelihood function $p(y_p|{\bm{x}},\Delta\bm{ \phi})=\mathcal{CN}(y_p;{\bm{f}}_p{\bm{x}},\sigma_e^2)$, where ${\bm{f}}_p$ is the $p$-th row of measurement matrix $\bm{F}$. The optimal off-grid offsets $\Delta\bm{\phi}$ can be obtained by a maximum likelihood (ML) approach
	\begin{equation}\label{eq_ML}
		\Delta\bm{\phi}^* = {\rm arg} \max \limits_{\Delta\bm{\phi}} {\rm ln}p({\bm{y}},\Delta\bm{ \phi})={\rm arg} \max \limits_{\Delta\bm{\phi}} {\rm ln} \int_{\bm{x}} p({\bm{x}},{\bm{y}},\Delta\bm{ \phi}){\rm d}{\bm{x}}
	\end{equation}
	
	After obtaining the estimated offsets $\Delta\bm{\phi}^*$, we can then calculate the corresponding marginal posterior distribution (\ref{eq_marginal_posterior}). It's the production of many distribution functions, each of which is determined by only a subset of variables. Such decomposition structures can be calculated by a graphical model, i.e., a factor graph, which can be an undirected bipartite graph used to represent the connection between random variables (called variable nodes) and related probability density functions (pdfs) (called factor nodes). Unfortunately, in (\ref{eq_marginal_posterior}), the generated factor graph contains a densely connected structure due to the measurement matrix, which {underestimates} the exact calculation of the posterior distribution using the factor graph. To handle this challenge, we propose a Recurrent-OAMP approach that can decouple the linear and nonlinear observations in the following subsection. 
	
	\subsection{Recurrent-OAMP Approach}
	The Recurrent-OAMP approach contains two iteratively working modules, as shown in Fig. \ref{fig_modules}: Module A is a linear MMSE (LMMSE) estimator, inputting observations and messages from Module B; Module B is an MMSE estimator, inputting the inherent sparse prior and messages from Module A. Both modules will work iteratively until convergence.
	\begin{figure}[!t]
		\begin{minipage}[t]{0.45\textwidth}
			\centering
			\includegraphics[width=3in]{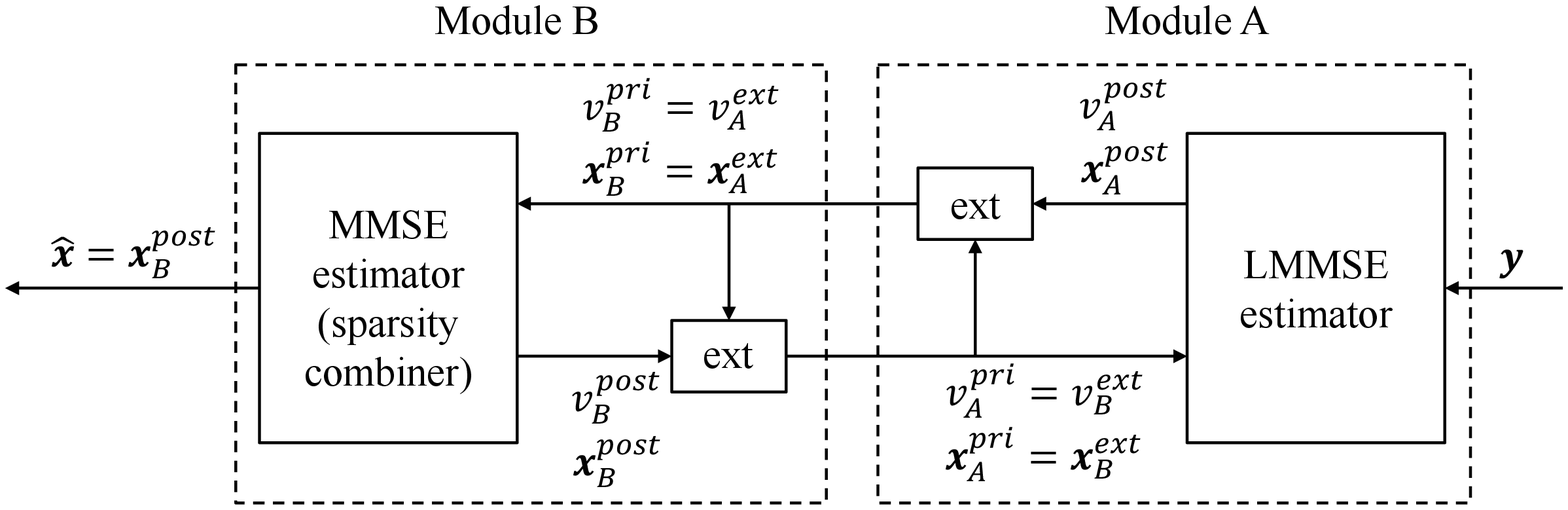}
			\caption{{Two recurrent modules of the R-OAMP approach.}}
			\label{fig_modules}
		\end{minipage}
		\hspace{0.02\linewidth}
		\begin{minipage}[t]{0.45\textwidth}
			\centering
			\includegraphics[width=3in,height=1.6in]{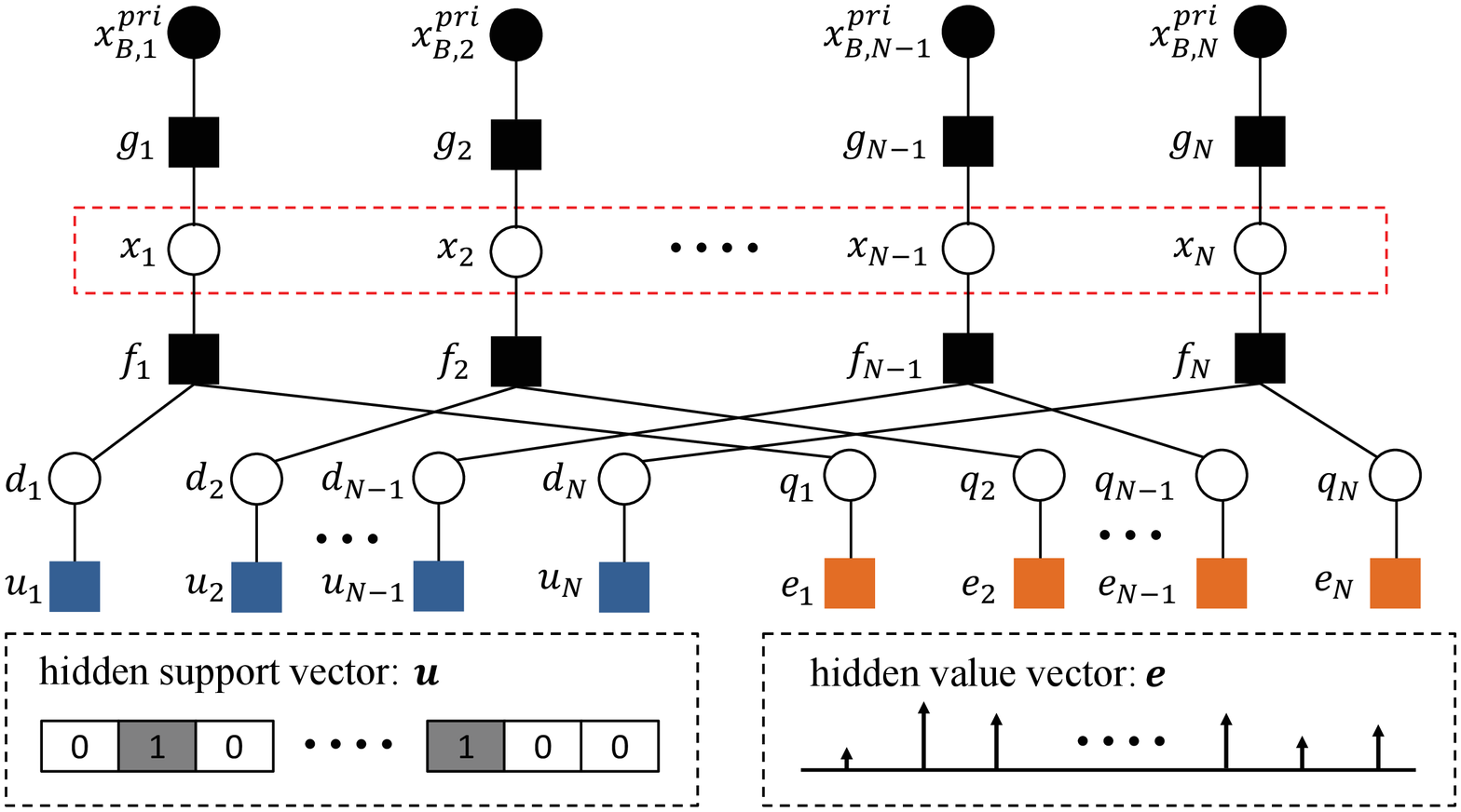}
			\caption{{A factor graph of hidden channel support and value vectors.}}
			\label{fig_factor_graph}
		\end{minipage}
	\end{figure}
	
	\subsubsection{ Module A: LMMSE Estimator}
	In Module A, since only linear observations are considered, we can consider that $\bm{x}$ has a Gaussian distribution with mean ${\bm{x}}^{pri}_A$ and variance $v_A^{pri}$ passed from module B without a sparse prior. Given received signals $\bm{y}$ and the prior distribution $\mathcal{CN}(\bm{x};{\bm{x}}^{pri}_A,v_A^{pri}\bm{I})$, the LMMSE estimation  ${\bm{x}}^{post}_A$ and its mean squared error (MSE) ${v_A^{post}}$ are respectively given by \cite{kay1993fundamentals}
\begin{equation}\label{eq_xapost_eq_vapost}
	\displaystyle {\bm{x}}^{post}_A={\bm{x}}^{pri}_A+\frac{v_A^{pri}}{v_A^{pri}+\sigma_e^2}{\bm{F}(\Delta\bm{\phi})}^{\rm H}({\bm{y}}-{\bm{F}(\Delta\bm{ \phi})}{\bm{x}}^{pri}_A), \text{ }
	\displaystyle {v_A^{post}}={v_A^{pri}}-\frac{P}{N}.\frac{(v_A^{pri})^2}{v_A^{pri}+\sigma_e^2}
\end{equation}
	
	As $\mathcal{CN}({\bm{x}};{\bm{x}}^{post}_A,v_A^{post}\bm{I})$ and $\mathcal{CN}({\bm{x}};{\bm{x}}^{pri}_A,v_A^{pri}\bm{I}) $ are independent of each other, the extrinsic message is also subject ot a Gaussian distribution and given by
	\begin{equation}\label{eq_}
		\mathcal{CN}({\bm{x}};{\bm{x}}^{post}_A,v_A^{post}\bm{I})\propto \mathcal{CN}({\bm{x}};{\bm{x}}^{pri}_A,v_A^{pri}\bm{I}) \mathcal{CN}({\bm{x}};{\bm{x}}^{ext}_A,v_A^{ext}\bm{I})
	\end{equation}
	
	Hence, the extrinsic mean and variance can be computed by \cite{kay1993fundamentals}
\begin{equation}\label{eq_x_aext_eq_v_aext}
	\displaystyle {\bm x}^{pri}_B={\bm x}^{ext}_A={v}^{ext}_A ({{\bm x}^{post}_A}/{{v}^{post}_A}-{{\bm x}^{pri}_A}/{{v}^{pri}_A}), \text{ }
	\displaystyle {v}^{pri}_B={v}^{ext}_A=({1}/{{v}^{post}_A}-{1}/{{v}^{pri}_A})^{-1}
\end{equation}
	where ${\bm x}^{pri}_B $ and ${v}^{pri}_B$ are the input mean and variance of Module B, which are equal to ${\bm x}^{ext}_A$ and ${v}^{ext}_A$, respectively. 
	\subsubsection{Module B: MMSE Estimator}
	In Module B, {the MMSE estimator is based on the assumption that ${\bm{x}}^{pri}_B={\bm{x}}+{\bm{z}}$ where ${\bm{z}} \sim \mathcal{CN}(0,{v}^{pri}_B\bm{I})$, which has been proved to be valid for OAMP \cite{DBLP:journals/access/MaP17}. Based on the assumption and the proposed sparsity channel priors, the factor graph of the joint distribution $p({\bm {x,d,q}{,{\pmb x}^{pri}_B}})$ denoted by $\mathcal{J}$ can be obtained, as shown in Fig. \ref{fig_factor_graph}. In the graph, the functional form of each factor node is listed as below:
	$g_n(x_n,x^{pri}_{B,n}) \sim p(x^{pri}_{B,n}|x_n)= \mathcal{CN}(x^{pri}_{B,n};{x_n},v_B^{pri})$, $f_n(x_n, d_n, q_n) \sim p(x_n|d_n, q_n)=\delta(x_n-d_nq_n)$, $u_n(d_n) \sim p(d_n)={(1-\lambda)^{1-d_n}(\lambda)^{d_n}}$, 
	$e_n(q_n) \sim p(q_n) = \mathcal{CN}({q_n};\zeta,\rho)$.
	Our goal is to calculate the marginal posterior distribution $\left\{p({x_n}|{\bm{x}}^{pri}_B)\right\}_{n=1}^{N}$ by directly exploring sum-product message passing rules on the graph $\mathcal{J}$. The details of the calculation of the Module B are presented in Algorithm 1.}

	When the posterior distribution is obtained by message passing in Module B, the extrinsic message will be sent back to module A. This message-passing pattern is similar to how a recurrent network works, so the recurrent-OAMP (R-OAMP) approach is named. R-OAMP executes message passing between Module A and Module B until convergence.
	Once the estimated posterior distribution $\hat{p}({\bm{x}|{\bm{y},\Delta\bm{ \phi}}})=\prod_{n=1}^N p(x_n|{\bm{y}},\Delta\bm{\phi})$ is obtained, we can update the off-grid offsets $\Delta\bm{\phi}$ using ML. Unfortunately, we can't obtain the closed-form expression of the log-likelihood function ${\rm ln}p({\bm{y}},\Delta\bm{\phi})$. As a result, the gradient descent method can't be directly utilized to derive the optimal solution $\Delta\bm{ \phi}^*$. 
 The following Theorem shows how to obtain $\Delta\bm{ \phi}$.
 \begin{thm}
     {\rm {After introducing a surrogate function $r(\Delta\bm{\phi};\Delta\bm {\dot\phi})=\int p({\bm{x}}|{\bm{y}},\Delta\bm{\dot \phi}){\rm ln}\frac{p({\bm{x}},{\bm{ y}},\Delta\bm{\phi})}{p({\bm{x}}|{\bm{y}},\Delta\bm{\dot \phi})}d{{\bm{x}}}$, one can explore a gradient ascend method to update $\Delta\bm{ \phi}$ with
	\begin{equation}\label{eq_update_phi}
		\displaystyle \Delta\bm{ \phi}^{i+1}= \Delta\bm{\phi}^{i} + \alpha_i\frac{\partial r(\Delta\bm{\phi};\Delta\bm{\phi}^i)}{\partial \Delta\bm{\phi}}|_{\Delta\bm{\phi}= \Delta\bm{\phi}^i}
	\end{equation}
	where $i$ represents the $i$-th iteration, $\Delta\bm{\phi}^{i}$ stands for the corresponding value of $\Delta\bm{\phi}$, and $\alpha_i$ is the update stepsize.
 }}
 \end{thm}
 \begin{proof}
		Please refer to Appendix B.
	\end{proof}
To obtain high-precision estimation results, we utilize the backtracking linear search method, which can avoid oscillating results or slow convergence by adaptively adjusting the stepsize.	
	Besides, please refer to {Appendix B} for the detailed update expression of the off-grid offsets $\Delta\bm{\phi}^{i+1}$. 
	
	\begin{algorithm}
\caption{Recurrent-OAMP}
\label{alg2}
\begin{algorithmic}[1]
\STATE \textbf{Input and initialization:}  Received signal (or observations) $\bm{y}$, measurement matrix ${\bm F}(0)$, noise variance $\sigma^2_e$, $I_c^{\max} = 50$, ${\bm{x}}^{pri}_A=\bm{0}$, ${v}^{pri}_A=\lambda\rho$, and $\Delta\bm{\phi}=\bm{0}$.
\WHILE{not converge}
\STATE \textbf{Module A (given  $\Delta\bm{\phi}=\Delta\bm{\phi}^i$):}
\STATE Update ${\rm \pmb x}^{post}_A$ and ${v}^{post}_A$ using (\ref{eq_xapost_eq_vapost}). Update ${\bm{x}}^{pri}_B={\bm{x}}^{ext}_A$ and ${v}^{pri}_B={v}^{ext}_A$ using (\ref{eq_x_aext_eq_v_aext}).
\STATE \textbf{Module B:}
\STATE Update $\eta_{g_n \rightarrow x_n}(x_n)=\mathcal{CN}({x_n};x^{pri}_{B,n},v_B^{pri})$ with outputs of Module A. 
\STATE Update ${\eta _{{f_n} \to {x_n}}}({x_n})= (1-\lambda)\delta(x_n)+\lambda \mathcal{CN}(x_n; \zeta,\rho)$.
\STATE Update $x^{post}_{B,n}={\rm E}({x_n}|{\bm{x}}^{pri}_B)$ and ${v}^{post}_B=\frac{1}{N}\sum_{n=1}^N{\rm Var}({x_n}|{\bm{x}}^{pri}_B)$, where the posterior distribution $p({x_n}|{\bm{x}}^{pri}_B)= \frac{\eta_{f_n \rightarrow x_n}(x_n)\eta_{g_n \rightarrow x_n}(x_n)}{(1-\lambda) \mathcal{CN}(0;x^{pri}_{B,n},v_B^{pri})+\lambda  \mathcal{CN}(0;x^{pri}_{B,n}-\zeta,v_B^{pri}+\rho)}$.
\STATE Update ${\bm{x}}^{pri}_A={\bm{x}}^{ext}_B={v}^{ext}_B ({{\bm{x}}^{post}_B}/{{v}^{post}_B}-{{\bm{x}}^{pri}_B}/{{v}^{pri}_B})$ and ${v}^{pri}_A={v}^{ext}_B=({1}/{{v}^{post}_B}-{1}/{{v}^{pri}_B})^{-1}$.
\STATE Repeat Module A and Module B until convergence.
\STATE Output $\hat{p}({\bm{x}|{ \bm{y},{\Delta\bm{\phi}}^i}})=\mathcal{CN}({x_n};{\bm{ x}}^{post}_{B},v_B^{post}\bm{I})$
\STATE Update off-grid offsets $\Delta\bm{\phi}^{i+1}$ using (\ref{eq_update_phi}). Update $i = i+ 1$.
\ENDWHILE
\STATE \textbf{Output:} $\hat{\bm{x}}={\bm{x}}_B^{post}$ and {$\hat {\bm h} = { \bm{A}(\Delta \bm{\phi}^{i+1}) \hat{\bm{x}}}$.}
\end{algorithmic}
\end{algorithm}
	
	\subsection{{Channel Enhancement via RIS Phase Alignment}}
	{In this subsection, we discuss how to enhance the HAP-UAV channel by aligning RIS phases.}
	The optimization of RIS phase shift is crucial for the performance improvement of RIS-integrated communication systems. By aligning the phase shift ${\bm{\Theta}}$, signals sent by the BS and then reflected by the RIS can be coherently added to the UAV, which thus significantly increase the HAP-UAV channel gain and {maximize} the received SNR by the UAV ${SNR}_u$. 
	
	Given the estimated HAP-UAV channel vector $\bm{ h}$, we can then formulate the RIS phase shift optimization problem as ${\bm \Theta}^{\star}  = \mathop {\arg \max }\limits_{\bm \Theta} \frac {P_bG}{\sigma_0^2}|\bm{ h}^{\rm H}{\bm \Theta}{\bm{H}}\bm{{v}}|^2 $, where $\sigma_0^2$ is the noise power. To achieve the optimal ${\bm \Theta}$, we first need to determine the precoding vector $\bm{v}$. The following Lemma presents a strategy for designing the precoding vector. 
	
	\begin{lemma}\label{lem_beamformer}
		{\rm For any given BS-HAP-UAV propagation channel and RIS phase shift, the BS needs to adopt a maximum ratio transmission (MRT) strategy to maximize ${SNR}_u$, that is, we have}
		\begin{equation}\label{eq_precoding_vec}
			\displaystyle \bm{v}=\frac{\bm{ a}_{BS}(\upsilon_{{\rm AoD}})}{||\bm{ a}_{BS}(\upsilon_{{\rm AoD}})||_2}
		\end{equation}
	\end{lemma}
	\begin{proof}
		Please refer to Appendix C.
	\end{proof}
	
	Using (\ref{eq_precoding_vec}), ${SNR}_u$ can be rewritten as
	\begin{equation}\label{eq_expand_SNRu}
		\begin{array}{l}
			SN{R_u} = \frac{{{P_b}G}}{{{\sigma_0 ^2}}}\left| {{{\bm h}^{\rm{H}}}{\bm \Theta} {\bm H}\frac{{{{\bm a}_{{\rm{BS}}}}({\upsilon _{{\rm{AoD}}}})}}{{||{{\bm a}_{{\rm{BS}}}}({\upsilon _{{\rm{AoD}}}})|{|_2}}}} \right|
			= \frac{{\alpha {P_b}G}}{{{\sigma_0 ^2}}}{\left| {\sum\limits_{n = 1}^N {\sum\limits_{l = 1}^L {{\beta _l}} } {e^{j({\theta _n} + 2\pi (n - 1){{\bar d}_{{\rm{RIS}}}}({\rm{sin}}({\omega _l}) - {\rm{sin}}({\upsilon _{{\rm{AoA}}}}))}}} \right|^2}
		\end{array}
	\end{equation}
	{where $\beta_l$ is the product of the $l$-th non-zero element of the estimated small-scale channel fading $\hat {\bm x}$ and the large-scale channel fading of the HAP-UAV channel}.
	
	From (\ref{eq_expand_SNRu}), we can find that ${SNR}_u$ is independent of AoD $\upsilon_{{\rm AoD}}$. However, as there are multiple reflected/scattered paths in the HAP-UAV EM propagation environment, it's difficult to optimize $\theta_n$ ($\forall n$) that can maximize the received signal power of the UAV. 
	{Hence, we attempt to obtain the suboptimal $\theta_n$}. 
 We arrange (\ref{eq_expand_SNRu}) as
	\begin{equation}\label{eq_rearranged_SNRu}
		\displaystyle {SNR}_u=\frac {\alpha P_bG}{\sigma_0^2}\left|\sum_{l=1}^{L}\beta_lr_l\right|^2 
	\end{equation}
	where $r_l \triangleq |\sum_{n=1}^{N} e^{j(\theta_n+2\pi (n-1)d_{\rm RIS}( {\rm sin}(\omega_{l})-{\rm sin}(\upsilon_{{\rm AoA}}))}|^2$.
 
	Maximizing (\ref{eq_rearranged_SNRu}) is equivalent to maximizing $\{r_l\}_{l \in {\mathcal L}}$. 
	Denote $l_0$ as the phase alignment path of RIS and $\psi_n^*(l_0)$ the corresponding phase shift. Then, we can coherently reflect signals via path $l_0$ by setting 
	\begin{equation}\label{eq_new_refl_path}
		\psi_n^*(l_0)=\theta-2\pi (n-1){\bar d}_{\rm RIS}( {\rm sin}(\omega_{l_0})-{\rm sin}(\upsilon_{{\rm AoA}}))
	\end{equation}
	where $\theta \in [0,2\pi)$ is an arbitrary phase shift. By enforcing 
    \begin{equation}\label{eq_new_refl_path_final_phi}
    \theta_n  =  \psi_n^*(l_0)
    \end{equation}
	and performing specific mathematical transformations, we can then rewrite $r_l$ as
	\begin{equation}\label{eq_simplified_rl}
		\displaystyle r_l=g(\Delta\psi_l(l_0)) \triangleq \left|\frac{{\rm sin}(\pi Nd_{\rm RIS}\Delta\psi_l(l_0))}{{\rm sin}(\pi d_{\rm RIS}\Delta\psi_l(l_0))}\right|^2
	\end{equation}
	with $\Delta\psi_l(l_0)={\rm sin}(\omega_l)-{\rm sin}(\omega_{l_0})$. 
 
 To maximize (\ref{eq_simplified_rl}), we need to minimize $|\Delta\psi_l(l_0)|$. Considering the low-latency requirement of URLLC services, we perform the first-order Taylor expansion on $|\Delta\psi_l(l_0)|$ to realize the rapid configuration of RIS phase shift, i.e., \begin{equation}\label{eq_}
		|\Delta\psi_l(l_0)| \approx |{\rm cos}(\omega_l)||\omega_l-\omega_{l_0}| 
	\end{equation}
	
	As a result, we can reformulate the RIS Phase shift optimization problem as
	\begin{subequations}\label{eq:optimize_RIS}
		\begin{alignat}{2}
			& \min \limits_{\omega_{l_0}, \zeta_l} \sum_{l=1}^L\left|{\beta_l}{{\rm cos}(\omega_l)}\right|\zeta_l   \\
			\allowdisplaybreaks[4]
			& {\rm s.t:} \text{ }
			 \omega_{l_0}-\zeta_l\leq \omega_l,\forall l \in {\mathcal L} \\ 
			& \omega_l \leq \omega_{l_0}+\zeta_l,\forall l \in {\mathcal L}
		\end{alignat}
	\end{subequations}

(\ref{eq:optimize_RIS}) is a linear programming problem and can then be efficiently solved by optimization tools such as MOSEK.
	
	\section{{RESOURCE OPTIMIZATION and ALGORITHM DESIGN}}
	We next discuss how to solve the original optimization problem with the {cascaded BS-HAP-UAV channel gain}. As mentioned above, the complicated DEP expression greatly hinders the theoretical analysis of the optimization problem. Therefore, we first attempt to approximate the DEP. 
	
	\subsection{DEP Approximation}
	To obtain an analysis-friendly expression of DEP ${\varepsilon _u}$, we improve the linear approximation method in \cite{DBLP:journals/tcom/DhokRSSL21}. The key idea of this method is to conduct the first-order Taylor series expansion at the point {${SNR}_u=\gamma_u$} to achieve an asymptotical expression. 
	{Different from the approximation in \cite{DBLP:journals/tcom/DhokRSSL21}, we extend the linear approximation region and propose to allocate redundant resources to achieve a lower DEP.} 
	Specifically, we approximate $\varepsilon_u$ as
	\begin{equation}\label{eq_}
		\begin{array}{l}
			{\varepsilon _u} \approx \Omega ({SNR}_u)
			= \left\{ {\begin{array}{*{20}{l}}
					{1,}&{SNR_u \le {{SNR}}_{\rm low}} \\
					{0,}&{SNR_u \ge {{SNR}}_{\rm up}} \\
					{\frac{1}{2} - \frac{1}{{2\chi _u}}({{SN}}{{{R}}_u} - {\gamma _u}).} & {\rm otherwise}
			\end{array}} \right.
		\end{array}
	\end{equation}
	{where $\displaystyle\gamma_u=2^{R_u}-1$, ${SNR}_{\rm low}={{\gamma_u} - \frac{1}{{{\chi _u}}}}$, ${SNR}_{\rm up} = {{\gamma_u} + \frac{1}{{{\chi _u}}}}$, and $\displaystyle \chi_u=\sqrt\frac{b_u}{2\pi}(2^{2R_u}-1)^{-1/2}$. }Similarly, $\Omega({SNR}_k)$ represents the Taylor series expansion based approximation of $Q(f(b_k,{SNR}_k))$.
	
	
	Recall that the UtG channel is a random channel. However, the UtG channel fading should be regarded as a constant during a transmission interval. Thus, we need to analyze the expected DEP, denoted by ${\overline \varepsilon  _k}$, experienced by any robot $k$, which can be calculated by
	\begin{equation}\label{eq_approximate_epsilon_aver}
		\begin{array}{l}
			{\overline \varepsilon  _k} = {\mathbb E}[Q(f({b_k},SN{R_k}))]
			\approx \int_0^{\infty}\Omega({SNR}_k)f_{{SNR}_k}(x){\rm d}x %
			=\chi_k\int_{\gamma_k-\frac{1}{{\chi_k}}}^{\gamma_k+\frac{1}{{\chi_k}}}F_{{SNR}_k}(x){\rm d}x
		\end{array}
	\end{equation}
	where $f_{{SNR}_k}(x)$ and $F_{{SNR}_k}(x)$ are the pdf and the cdf of ${SNR}_k$, respectively. 
 
	As the Rayleigh small-scale fading is considered, ${SNR}_k$ follows an exponential distribution. Thus, $F_{{SNR}_k}(x)$ can be computed by {\cite{DBLP:journals/wcl/GuCLV18}}
	\begin{equation}\label{eq_F_SNR}
		F_{{SNR}_k}(x) = \begin{cases}1-e^{-\frac{x}{\overline{{SNR}}_k}}, & x > 0 \\0,&  x\leq 0\end{cases}
	\end{equation}
	where $\overline{{SNR}}_k=P_u|g_k^L|^2/\sigma_{k}^2$ is the average value of ${SNR}_k$. 
 
 By substituting (\ref{eq_F_SNR}) into (\ref{eq_approximate_epsilon_aver}), for any robot $k \in {\mathcal K}$, we can calculate its expected DEP by
\begin{equation}\label{eq_approx_epsilon}
		\begin{array}{l}
			{\overline \varepsilon  _k} \approx {\chi _k}\int_{{\gamma _k} - \frac{1}{{{\chi _k}}}}^{{\gamma _k} + \frac{1}{{{\chi _k}}}} {(1 - {e^{ - x/{{\overline {SNR} }_k}}})} {\rm d}x  			{\mathop  \approx \limits^{(a)} \frac{2\gamma_k}{{\overline{SNR}}_k}}
		\end{array}
\end{equation}
	where (a) follows from the approximation $1-e^{-x/\overline{{SNR}}_k}\approx \frac{x}{\overline{{SNR}}_k}$ under an assumption of high SNR \cite{DBLP:journals/wcl/GuCLV18}. This assumption is reasonable for URLLC applications because robots need to achieve high SNR such that the stringent QoS requirements can be satisfied. 
	
	
	\subsection{Original Problem Decomposition}
	After obtaining the estimated HAP-UAV channel gain, the aligned RIS phase, and the approximated DEPs, the original problem (\ref{eq:formualted_problem}) becomes analytically tractable. 
	Nevertheless, we observe that there are integer and continuous variables in (\ref{eq:formualted_problem}) and the decision variables are complexly coupled. (\ref{eq:formualted_problem}) can be confirmed to be a MINCP problem, which is challenging. 
	Generally, resorting to an iterative optimization scheme is a feasible way of solving this challenging problem. 
	Yet, the developed algorithm would contain multiple nested loops if we directly adopt the iterative optimization scheme, which increases the computational complexity of the algorithm. 
	
	Fortunately, given a BS blocklength $b$, we observe that the optimal $P_b$ does not vary with the UAV transmit power $P_u$. 
	This is because there is a narrow feasible region $\hat {\mathcal D}_p \subseteq  {\mathcal D_p} = \left\{ {{P_b}\left| {Q\left( {f\left( {{b_u},SN{R_u}} \right)} \right) \le \varepsilon _u^{{\rm{th}}},0 < {P_b} \le {P_B}} \right.} \right\}$, in which $\Omega(SNR_u)$ strictly monotonically decreases with $P_b$. 
	Given the blocklength $b$, the objective function (\ref{eq:formualted_problem}a) then reduces to a monotonically increasing function over $P_b$ in the specific feasible region $\hat {\mathcal D}_p$. 
	
	Therefore, we can apply horizontal decomposition to (\ref{eq:formualted_problem}) that creates a problem of two-layered structure. The $1^{\rm st}$ layer is \emph{BS-Layer Optimization}, and the $2^{\rm nd}$ layer is \emph{UAV-Layer Optimization}. More importantly, we need not to conduct iterative optimization between these two layers.
	
	
	\subsubsection{BS-Layer Optimization}
	Given any decision variables-related to the UAV-Layer Optimization, we can formulate the BS-Layer Optimization problem as 
	\begin{subequations}\label{eq:BS_layer_opt}
		\begin{alignat}{2}
			& \max \limits_{P_b, {b_u}} \text{ }  \min \limits_{k \in {\mathcal K}} \text{ } \displaystyle  \frac{R_u(1-\epsilon_u)+R_k(1-\epsilon_k)}{{\bar p}_b+{\bar p}_u}   \\
			\allowdisplaybreaks[4]
			& {\rm s.t:} \text{ }
			 \Omega(SNR_u) \le \varepsilon_u^{\rm th}, B_{u}^{\min} \leq b_u \leq B_{u}^{\max}, \text{ } b_u \in \mathbb{Z^+}, 0 < P_b\leq P_{B}
		\end{alignat}
	\end{subequations}
	where $\min  \left \{ x_k \right \}$ denotes the minimum $x_k$, $\forall k \in {{\mathcal K}}$.
	
	We observe from the expressions of $R_u$ and $\Omega(SNR_u)$ that the decision variables $P_b$ and $b_u$ in (\ref{eq:BS_layer_opt}) are intricately coupled, which challenges the solution to (\ref{eq:BS_layer_opt}). The following Lemma presents an effective method of solving it.
 \begin{lemma}\label{cal_pb_bu}
     {{\rm One can obtain the optimum value of (\ref{eq:BS_layer_opt}) by alternately computing
     \begin{equation}\label{eq_optimal_BS_power}
		\displaystyle  P_b(b_u)={\rm \min} \left\{P_B, {SNR_{\rm up}}/{\Delta B} \right\}
	\end{equation}
 and solving (\ref{eq:BS_blocklength_opt}) using an exhaustively searching method.
 {\begin{equation}\label{eq:BS_blocklength_opt}
			\max \limits_{b_u}  \text{ } \min \limits_{k \in {\mathcal K}} \displaystyle \frac{R_u(1-\epsilon_u)+R_k(1-\epsilon_k)}{P_b(b_B)/P_B+\overline{p}_u}  
		\text{ } {\rm s.t:} \text{ }
    B_{u}^{\min} \leq b_u \leq B_{u}^{\max}, \text{ } b_u \in \mathbb{Z^+},\text{ } \Omega(SNR_u) \le \varepsilon_u^{\rm th}
\end{equation}}
}}
{\rm {Further, the optimum value will be obtained in two iterations.}
}
 \end{lemma}
 \begin{proof}
		Please refer to Appendix D.
	\end{proof}

	\subsubsection{UAV-Layer Optimization}
	Given any decision variables related to the BS-Layer Optimization, we can formulate the UAV-Layer Optimization problem as 
	\begin{subequations}\label{eq:UAV_layer_opt}
		\begin{alignat}{2}
			&  \max \limits_{P_u, \{b_k\}} \text{ } \min  \limits_{k \in {\mathcal K}} \text{ } \displaystyle  \frac{R_u(1-\epsilon_u)+R_k(1-\epsilon_k)}{{\bar p}_b+{\bar p}_u}  \\
			& {\rm s.t:} \text{ }
			   {{2\gamma_k}}/{{\overline{SNR}}_k} \le \varepsilon_k^{\rm th},\text{ } \forall k \in {\mathcal K},  \text{ } B^{\min}\leq b_k \leq B^{\max}, \text{ } b_k  \in \mathbb{Z^+},\forall k\in {\mathcal K}, \text{ } 0 < P_u \leq P_{U}
		\end{alignat}
	\end{subequations}
	
	As (\ref{eq:UAV_layer_opt}) is a mixed-integer optimization problem,  {we decompose it into two iteratively optimized subproblems}, i.e., \emph{UAV transmit power control} subproblem and \emph{UAV blocklength optimization} subproblem, such that (\ref{eq:UAV_layer_opt}) can be effectively solved. 
	
	
	\textbf{2-a) UAV transmit power control:}
	Assuming that the values of the blocklength of $b_k$, $\forall k$, ${P}_b$, and $b_u$ are given, we can then formulate the UAV transmit power control subproblem as
	\begin{equation}\label{eq:UAV_power_control}
			  \max \limits_{P_u} \text{ } \eta_{\rm EE} = \min  \limits_{k \in {\mathcal K}} \text{ }  \displaystyle  \frac{R_u(1-\epsilon_u)+R_k(1-\epsilon_k)}{{\bar p}_b+{\bar p}_u}  
			\text{ } {\rm s.t:} \text{ }
   {{2\gamma_k}}/{{\overline{SNR}}_k} \le \varepsilon_k^{\rm th}, \text{ } 0 < P_u \leq P_{U}
	\end{equation}
	
	The objective function of (\ref{eq:UAV_power_control}) is a complex fraction, which makes it hard to be directly optimized. The following Lemma shows that (\ref{eq:UAV_power_control}) is equivalent to (\ref{eq:equi_reform_UAV_power_control}).
 \begin{lemma}\label{lem_UAV_transmit_power}
     \rm{One can obtain the optimum of (\ref{eq:UAV_power_control}) by solving the following rotated quadratic cone programming problem 
\begin{subequations}\label{eq:equi_reform_UAV_power_control}
		\begin{alignat}{2}
			&  \max \limits_{P_u,y} \text{ } y  \\
			\allowdisplaybreaks[4]
			& {\rm s.t:} \text{ }
			 \displaystyle {{2\gamma_k} \sigma_{k}^2}/({|g_k^L|^2 {\bar p}_uP_U}) \leq 1+({R_u(1-\epsilon_u)-y-\eta_{\rm EE}({{\bar p}_b+{\bar p}_u})})/{R_k},\forall k\in {\mathcal K} \\
            & \displaystyle {\bar p}_u \geq {{2\gamma_k} \sigma_{k}^2}/({|g_k^L|^2 \varepsilon^{\rm th}_kP_U}),\forall k\in {\mathcal K} \\ 
			& 0 < P_u \leq P_{U}
		\end{alignat}
	\end{subequations}
 }
 \end{lemma}
 \begin{proof}
		Please refer to Appendix E.
	\end{proof}
Now, (\ref{eq:equi_reform_UAV_power_control}) can be efficiently optimized by some commercial optimization tools such as MOSEK. 
	Further, we can conclude the steps of optimizing (\ref{eq:UAV_power_control}) in the Algorithm \ref{alg_UAV_power_control}. 
	\begin{algorithm}
		\caption{UAV Transmit Power Control}
		\label{alg_UAV_power_control}
		\begin{algorithmic}[1]
			\STATE \textbf{Initialization:} Initialize $\eta_{\rm EE}^{{(0)}}$ and $y^{(0)}$. Let error tolerance parameter $\Delta_r = 1{\rm e}-3$, maximum iterations $r_{\rm max} = 50$, and $r = 0$. Input $\{b_k\}$ and the optimized $b_u$ and $P_b$. 
			\REPEAT
			\STATE Given $\eta_{\rm EE}^{{(r)}}$, solve (\ref{eq:equi_reform_UAV_power_control}) to obtain the optimized UAV transmit power ${P}_u^{(r+1)}$ and $y^{(r+1)}$ using the MOSEK tool.
		
		    \STATE Given ${P}_u^{(r+1)}$, update $\eta_{\rm EE}^{{(r+1)}}$ by
		    \begin{equation}\label{eq_}
		    	\displaystyle  \eta_{\rm EE}^{{(r+1)}}={\min \limits_{k \in {\mathcal K}} \left\{ R_u(1-\epsilon_u)+R_k(1-{{2\gamma_k} \sigma_{k}^2}/({|g_k^L|^2 {\bar p}_u^{(r+1)}P_U}))) \right\}}/({{\bar p}_b+{\bar p}_u^{(r+1)}})
		    \end{equation}
			\UNTIL{The convergence criterion $|y^{(r+1)}-y^{(r)}|\leq \Delta_r$ is met or reach the maximum number of iterations $r_{\rm max}$.}
		\end{algorithmic}
	\end{algorithm}
	
	\textbf{2-b) UAV blocklength optimization}
	Given the UAV transmit power $P_u$, we can formulate the UAV blocklength optimization subproblem as
	\begin{subequations}\label{eq:UAV_block_opt}
		\begin{alignat}{2}
			&  \max_{\{b_k\}} \min  \limits_{k \in {\mathcal K}} \text{ }   R_k(1-\epsilon_k) 
			\text{ } {\rm s.t:} \text{ }
            B^{\min}\leq b_k \leq B^{\max}, \text{ } b_k  \in \mathbb{Z^+},\text{ } {{2\gamma_k}}/{{\overline{SNR}}_k} \le \varepsilon_k^{\rm th},\text{ } \forall k \in {\mathcal K}
		\end{alignat}
	\end{subequations}
		
	Note that (\ref{eq:UAV_block_opt}) is a {non-linear} integer programming problem. The computational complexity of solving it using some traditional optimization methods is extremely high. 
	Nevertheless, the feasible region of (\ref{eq:UAV_block_opt}) is small. Then, we can resort to the exhaustive search to achieve its optimal solution quickly. Moreover, we can search $\{b_k\}_{k=1}^{K}$ in parallel.
	
	
	\subsection{{Algorithm Design}}
 
	Finally, according to the above analysis and derivation, 
	we can summarize the main steps of the joint RIS Phase shift, Transmit Power, and Blocklength optimization (PTPB) algorithm for mitigating the formulated energy efficiency optimization problem in Algorithm \ref{alg_PTPB}. 
 
\begin{algorithm}
	\caption{Joint RIS Phase shift, Transmit Power, and Blocklength optimization, PTPB}
	\label{alg_PTPB}
	\begin{algorithmic}[1]
		\STATE \textbf{Initialization:} Run the initialization steps in R-OAMP and Algorithm \ref{alg_UAV_power_control}. Let $c_{\rm max} = 100$, $r = 0$. $I_{\rm max} = 10$.
		\STATE \textbf{Channel estimation and RIS phase shift optimization:}
		\STATE Call the channel estimation approach (R-OAMP) to estimate the HAP-UAV channel gain $\bm h$. 
		\STATE Compute $\bm v$ by (\ref{eq_precoding_vec}) and optimize (\ref{eq:optimize_RIS}) using MOSEK to obtain $\omega_{l_0}$. Compute $\theta_n$ by (\ref{eq_new_refl_path_final_phi}). 
		\STATE Given the obtained $\bm h$, $\bm v$, and $\bm \Theta$, compute $SNR_u$ .
		\STATE \textbf{BS-Layer optimization:}
		\REPEAT 
		\STATE Given $b_u^{(i)}$ compute $P_b(b_u^{(i)})^{(i+1)}$ by (\ref{eq_optimal_BS_power}).
		\STATE Given $P_b(b_u^{(i)})^{(i+1)}$, and any non-zero $P_U$ and $\{b_k\}$, exhaustively search $b_u^{(i+1)}$ such that (\ref{eq:BS_blocklength_opt}a) is maximized, subject to (\ref{eq:BS_blocklength_opt}b). Update $i = i+ 1$.
		\UNTIL{Converge or reach the maximum number of iterations $I_{\rm max}$.}
		\STATE \textbf{UAV-Layer optimization:}
		\REPEAT
		\STATE Call Algorithm \ref{alg_UAV_power_control} to obtain $P_u^{(c+1)}$. 
		\FOR{each $k \in \{1, 2, \ldots, K\}$ in parallel}
        \STATE Given $P_u^{(c+1)}$, solve the following problem using an exhaustive search method
        \begin{equation}\label{eq:BS_block_opt_k}
			\max_{{b_k}}  \text{ }   R_k(1-\epsilon_k) \text{ }
			 {\rm s.t:} \text{ }
			 {{2\gamma_k}}/{{\overline{SNR}}_k} \le \varepsilon_k^{\rm th},
			\text{ } B^{\min}\leq b_k \leq B^{\max}, \text{ } b_k  \in \mathbb{Z^+}
	\end{equation}
        \ENDFOR
        \STATE Update $c = c+1$.
		\UNTIL{Converge or $c = c_{\rm max}$.}
	\end{algorithmic}
\end{algorithm}

{The convergence analysis of Algorithm \ref{alg_PTPB} can be decomposed into the convergence analysis of the off-grid parameter updating module, the \emph{BS-Layer optimization} module, and the \emph{UAV-Layer optimization} module. 
In the off-grid parameter updating module, we utilize the in-exact majorization-minimization (MM) approach to approximate the optimal off-grids $\Delta\bm{ \phi}^*$, which guarantees to converge to a stationary point of (\ref{eq_ML}) \cite{wu1983convergence}. 
As presented in Lemma 2, we can obtain the optimum value of the \emph{BS-Layer optimization} problem in two iterations.  
Besides, in the \emph{UAV-Layer optimization} module, we adopt an iterative optimization scheme to solve the \emph{UAV-Layer optimization} problem. The reader can refer to \cite{DBLP:journals/jsac/YangXGQCC21} for a similar convergence proof of the iterative optimization scheme, and we omit the proof here for brevity.}

The computational complexity of Algorithm \ref{alg_PTPB} consists of three contributors including that of the channel estimation module, the RIS phase shift optimization module, and the joint power and blocklength optimization module. And the total complexity is $O(I^{\rm max}_c(PN^2+PN+N) + I^{\rm max}(S_1(B^{\rm max}-B^{\rm min})+\log_2(B^{\rm max}-B^{\rm min})) + (1+L)^{3.5} + c^{\rm max}(r^{\rm max}2^{3.5}+S_2(B_u^{\rm max}-B_u^{\rm min})+\log_2(B_u^{\rm max}-B_u^{\rm min}))) $ in the worse-case \cite{DBLP:books/daglib/0094154}, where $S_1$ is the complexity of computing (\ref{eq:BS_blocklength_opt}a) and $S_2$ is the complexity of computing (\ref{eq:UAV_block_opt}). Please refer to Appendix F for the detailed discussion on the algorithm complexity.

	\section{Simulation results}
	\subsection{Comparison Algorithms and Simulation Parameters}
	{In this section, we evaluate the performance of the proposed channel estimation approach, RIS
		phase shift optimization strategy, and resource optimization algorithm, respectively.}
	For this aim, the following baseline algorithms are considered: 
		 \textbf{OMP} \cite{karabulut2004sparse}: The OMP approach is an improved version of the well-known matching pursuit approach.
		 \textbf{SBL} \cite{DBLP:journals/tsp/DaiLL18}: The sparse Bayesian learning (SBL) approach recovers sparse signals based on a sparse Bayesian learning theory.
		 \textbf{SP} \cite{DBLP:journals/tit/DaiM09a}: The subspace pursuit (SP) approach identifies the signal support based on the maximum correlation criterion and is robust to measurement noises.
		 \textbf{Random-Phase}: In this algorithm, each element of the HAP-RIS is assigned with a random phase shift.
		 \textbf{Zero-Phase}: In this algorithm, all elements of the HAP-RIS are assigned with a zero phase shift.
		 {\textbf{Exhaustive-Phase}: It exhaustively searches for the AoD of the
		optimal phase alignment path ${\omega_{l_0}}$ in the interval $[0, \pi/2]$.}
		 \textbf{MTP}: The difference between the maximum transmit power-based (MTP) algorithm and the proposed algorithm lies in that MTP adopts the maximum BS transmit power and the maximum UAV transmit power.
		 \textbf{MBL}: The difference between the maximum blocklength-based (MBL) algorithm and the proposed algorithm lies in that MBL is configured with the maximum BS blocklength and the maximum UAV blocklength. 
	
	We consider a square area $\mathcal{G}$ of size $500 \times 500$ m$^2$, the location of the center point of which is $[80 \ {\rm km}, 0 \ {\rm km}]$. The robots are randomly distributed in $\mathcal{G}$. 
  {Both path loss of the BS-HAP channel and large-scale channel fading of the HAP-UAV channel are generated by the propagation model in{\cite{DBLP:journals/twc/AlsharoaA20}}. 
	The HAP-UAV small-scale channel realization consists of $L = 8$ propagation paths to simulate the sparsity characteristic, 
	and AoDs $\{\omega_l\}$ are assumed to follow an i.i.d. uniform distribution in an angular interval of ${\pi/12}$ centered on the LoS direction between the HAP and the UAV due to the limited scattering environment around the UAV. Considering the requirements of HAP-UAV channel sparsity along with BS-HAP-UAV long-distance propagation, 
 we set the carrier frequency $f_{\rm BS}=6$ GHz \cite{DBLP:journals/corr/abs-2104-01723}. As for the UtG channel model, we assume that the channel fading consists of large–scale and small-scale channel fading \cite{DBLP:journals/comsur/KhuwajaCZAD18}. We leverage the channel model in \cite{DBLP:journals/wcl/Al-HouraniSL14} and thereof channel parameters of a suburban scenario to calculate the UtG large-scale channel fading.
    The rayleigh fading model is utilized to model UtG small-scale channel fading{\footnote{{As the UtG channel is not sparse\cite{DBLP:journals/comsur/KhuwajaCZAD18}, we cannot leverage the proposed channel estimation approach to estimate its small-scale fading.}}}.
	The inter-element spacing of the RIS is assumed to be half-wavelength. More system parameters are listed as below:
	$M=32$, $N=128$, ${\bm p}_H=[1 ,0, 18]^{\rm T}$ km, $K=10$,	${\bm p}_U=[80, 0, 0.05]^{\rm T}$ km, $F_B=80$ bits, $B^{\min}_u=100$ bits, $B^{\min}=100$ bits, $B^{\max}_u=1000$ bits, $B^{\max}=1000$ bits, $P_B=120$ W, $P_U=0.5$ W, $f_{\rm UAV}=2$ GHz \cite{DBLP:journals/corr/abs-2104-01723}, $\varepsilon^{\rm th}_u=0.00005$, $\varepsilon^{\rm th}_k=0.00005$, 
 $G=4$ dB, $\sigma_0^2=-134$ dBm, $\sigma^2_k=-143$ dBm, $\zeta=0$, $\rho=1$, $ F_k=80$ bits \cite{DBLP:journals/icl/PanRDEN19}.}
	
	\subsection{Performance Evaluation}
	{We design extensive simulations to} comprehensively evaluate the performance of the proposed algorithm, particularly including the verification of the accuracy of the proposed channel estimation approach and the superiority of the proposed resource allocation algorithm.  
 All comparison approaches and algorithms in the simulations are executed for 200 times, and the final simulation result is the averaged one. 
	
	\subsubsection{Results of Channel Estimation}
	In this simulation, we compare the proposed R-OAMP approach with other channel estimation approaches, including OMP, SBL, and SP. The parameters $\left\{ {\lambda ,\varsigma ,\rho } \right\}$ are automatically updated by adopting an expectation-maximization algorithm during each iteration {\cite{DBLP:journals/tsp/ZinielS13a}}. 
	The metric, normalized mean square error (NMSE), is selected to evaluate the performance of channel estimation approaches, which is defined as
	\begin{equation}
		NMSE(\hat {\bm h}) \buildrel \Delta \over = {{{{\| {\hat {\bm h} - {\bm h}} \|}^2}}}/{{{{\left\| {\bm h} \right\|}^2}}}
	\end{equation}
	where ${\hat {\bm h}}$ is the estimation of ${\bm h}$.
	
	We plot the tendency of estimation accuracy of all compared channel estimation approaches under different $SNR_u \in \{0,4,8,12,16,20\}$, as shown in Fig. \ref{fig_NMSE}. From this figure, we can observe that: 1) For all channel estimation approaches, their estimation performance will be improved with an increasing $SNR_u$. 2) Compared with other estimation approaches, the estimation accuracy is improved by at least {$1.46$ dB} using R-OAMP. 3) When there are more RIS reflecting elements, which indicates that the HAP-UAV channel becomes sparser, the estimation accuracy of the proposed approach is further improved. 
	\begin{figure}[!t]
		\centering
		\subfigure[$N=128$]{\includegraphics[width=2.4in]{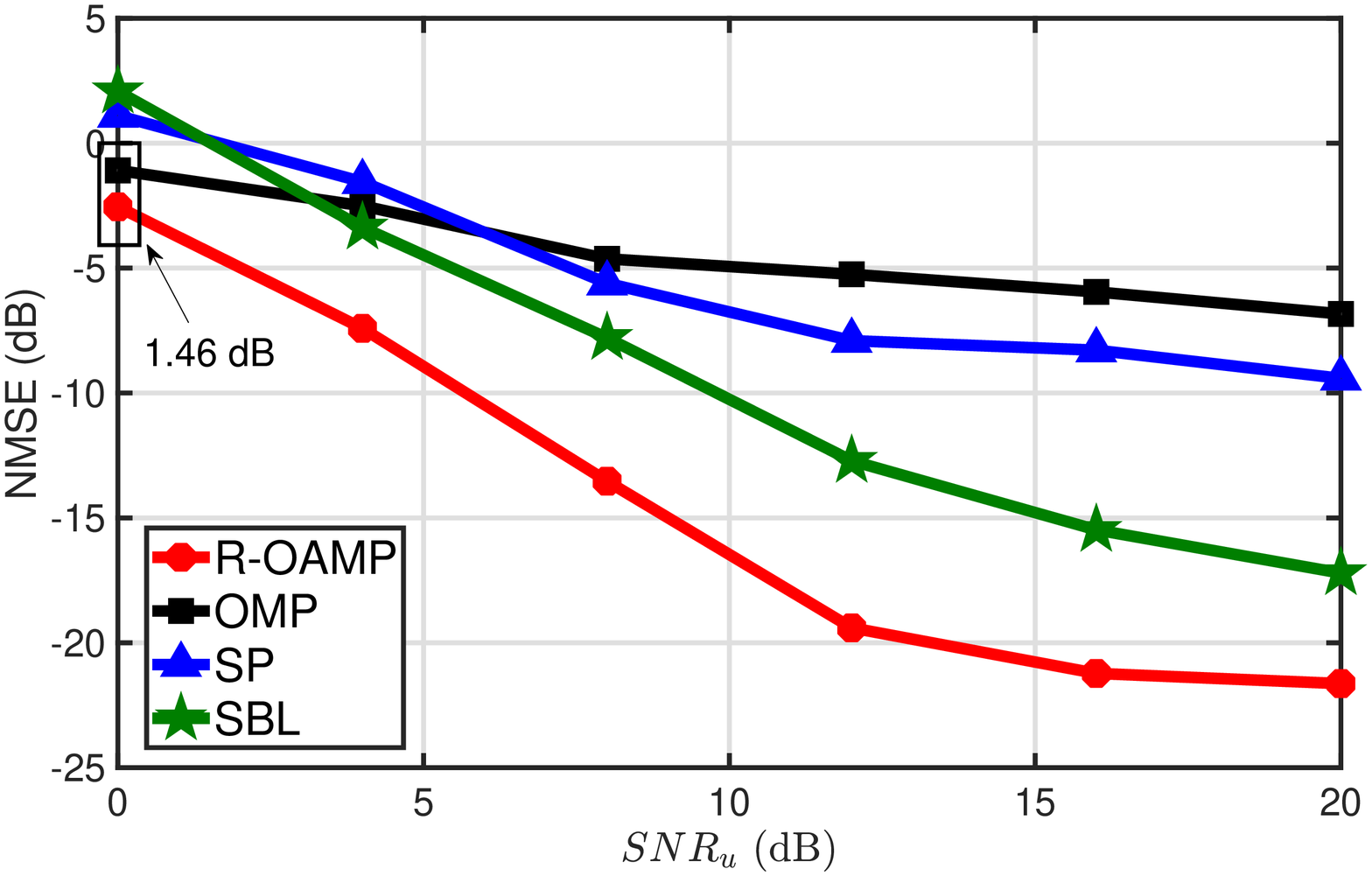}%
			\label{fig_NMSE_first_case}}
		\hfil
		\subfigure[$N=256$]{\includegraphics[width=2.4in]{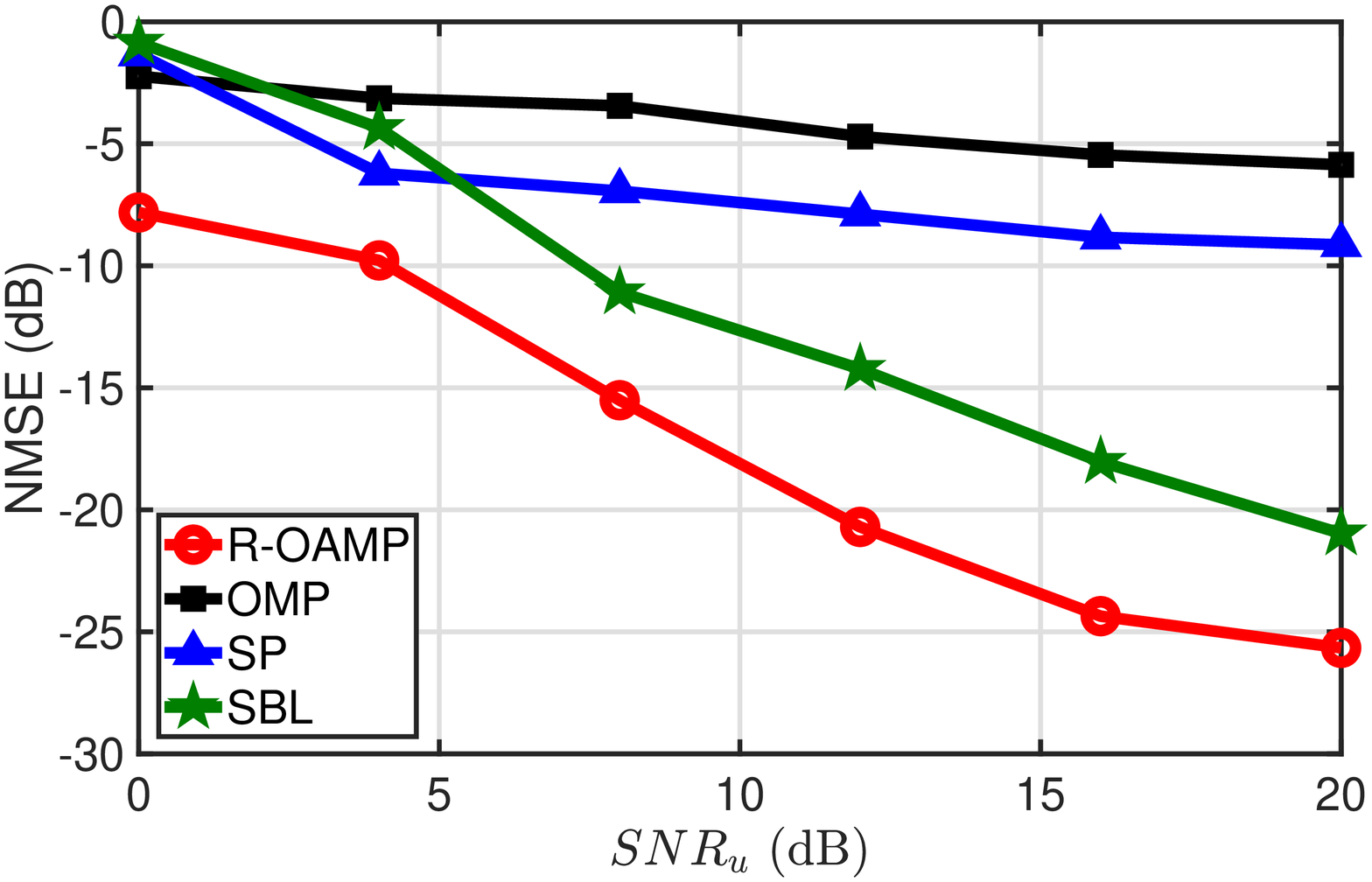}%
			\label{fig_NMSE_second_case}}
		\caption{NMSE versus received SNR of UAV ${\rm SNR}_{u}$, $P=48$.}
		\label{fig_NMSE}
	\end{figure}
	
	The estimation performance of all compared channel estimation approaches versus the number of pilot sequences $P$ is illustrated in Fig. \ref{fig_NMSE_versus_pilot}. We can obtain the following observations from this figure: 1) Except for OMP, the obtained NMSE of the rest estimation approaches decreases sharply with an increasing number of pilot sequences. For example, the obtained NMSE by R-OAMP decreases by {$15.24$ dB} when $P$ increases from $30$ to $70$. 2) R-OAMP can always achieve the smallest NMSE under different $P$ and $N$. It demonstrates that R-OAMP can achieve more accurate channel estimation results with lower pilot overhead by exploiting the sparsity prior of the HAP-UAV channel in the angular domain. 
	\begin{figure}[!t]
		\centering
		\subfigure[$N=128$]{\includegraphics[width=2.4in]{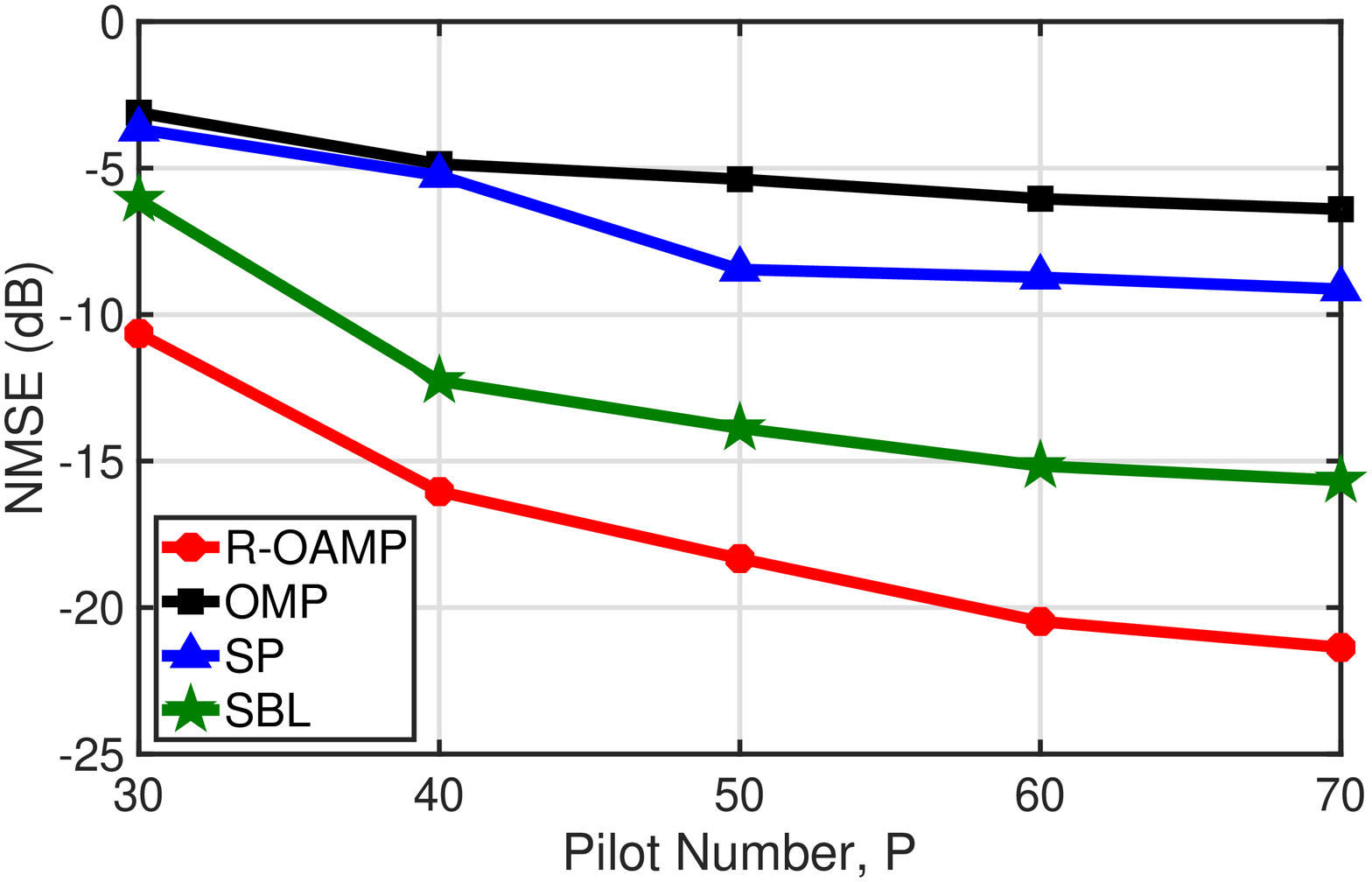}%
			\label{fig_NMSE_first_case_p}}
		\hfil
		\subfigure[$N=256$]{\includegraphics[width=2.4in]{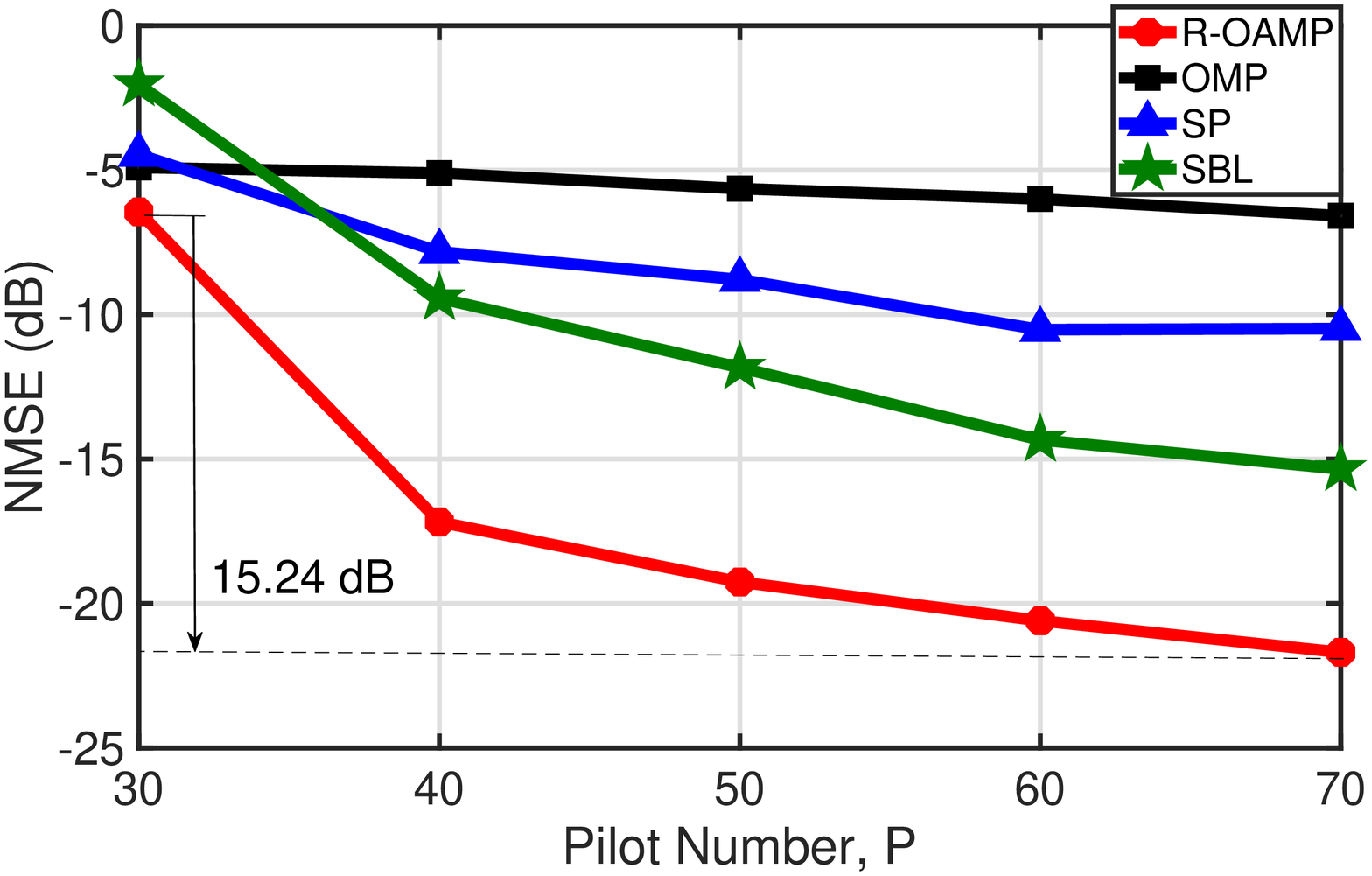}%
			\label{fig_NMSE_second_case_p}}
		\caption{NMSE versus the number of pilot sequences $P$, ${\rm SNR}_{u}=15 \rm dB$.}
		\label{fig_NMSE_versus_pilot}
	\end{figure}
	
	\subsubsection{Results of RIS Phase Shift Optimization}
	In this simulation, we compare the proposed algorithm with two benchmarks, i.e., Random-Phase and Zero-Phase, to verify its effectiveness regarding RIS phase shift optimization. 
	
	The tendency of the cascaded BS-HAP-UAV channel gain versus the number of RIS passive reflecting elements {$N \in \{96, 128, 160, 192, 224, 256\}$} is plotted in Fig. \ref{fig_ris_channel_gain}. We can observe from this figure that: 1) PTPB outweighs the benchmarks and improves the BS-HAP-UAV channel gain by at least {$20.36$} dB. {2) PTPB achieves a cascaded channel gain that is close to the channel gain obtained by the Exhaustive-Phase algorithm, and their difference is less than $3$ dB.}
	3) The obtained channel gain by PTPB monotonically increases with the number of RIS elements, while the benchmarks fail to unlock the significant advantage of RIS in terms of addressing the channel fading issue. 
	
	\begin{figure}[!t]
		\begin{minipage}[t]{0.45\textwidth}
			\centering
			\includegraphics[width=2.4in]{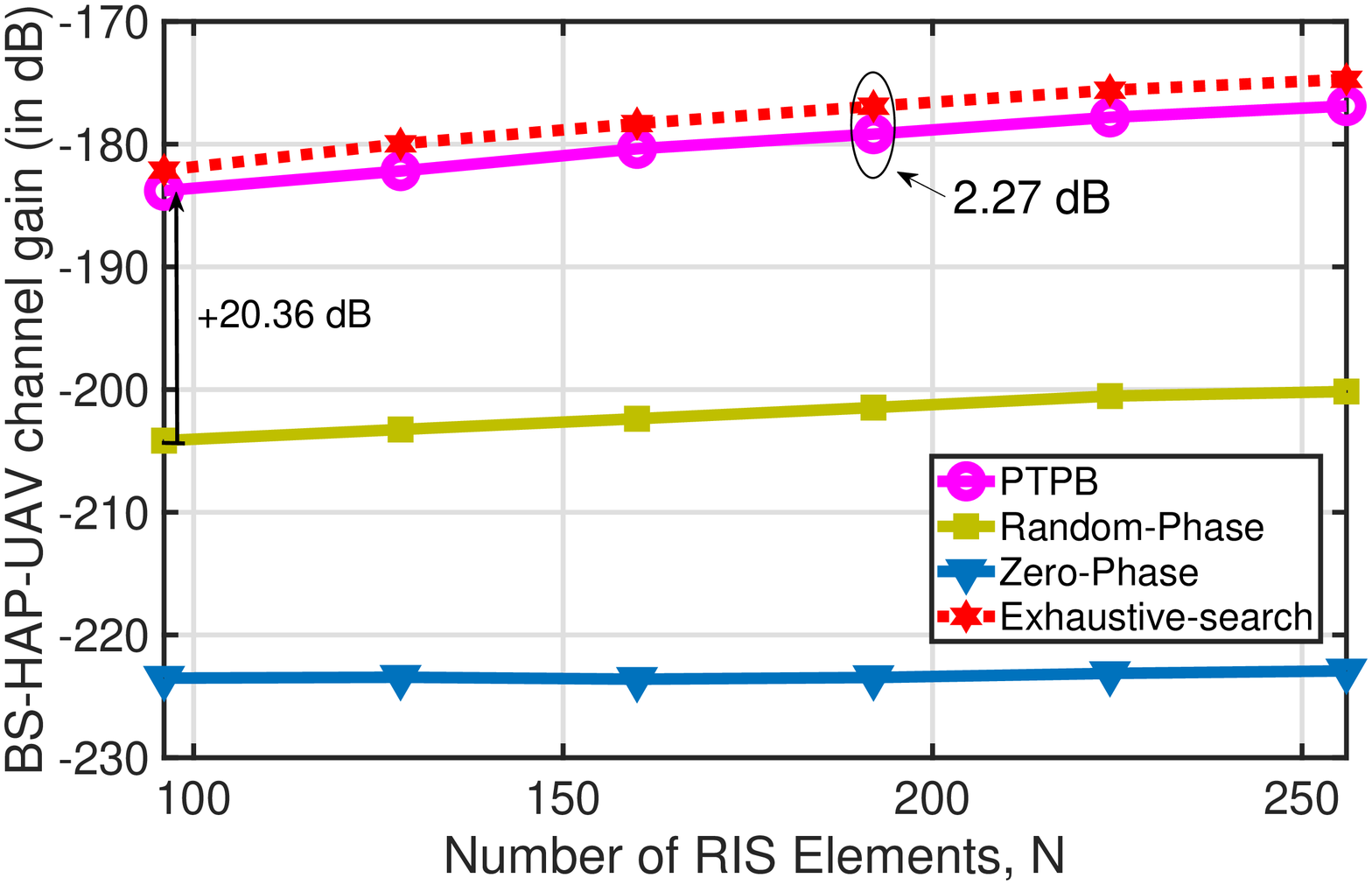}
			\caption{{BS-HAP-UAV channel gain versus the number of RIS elements $N$.}}
			\label{fig_ris_channel_gain}
		\end{minipage}
		\hspace{0.05\linewidth}
		\begin{minipage}[t]{0.45\textwidth}
			\centering
			\includegraphics[width=2.4in]{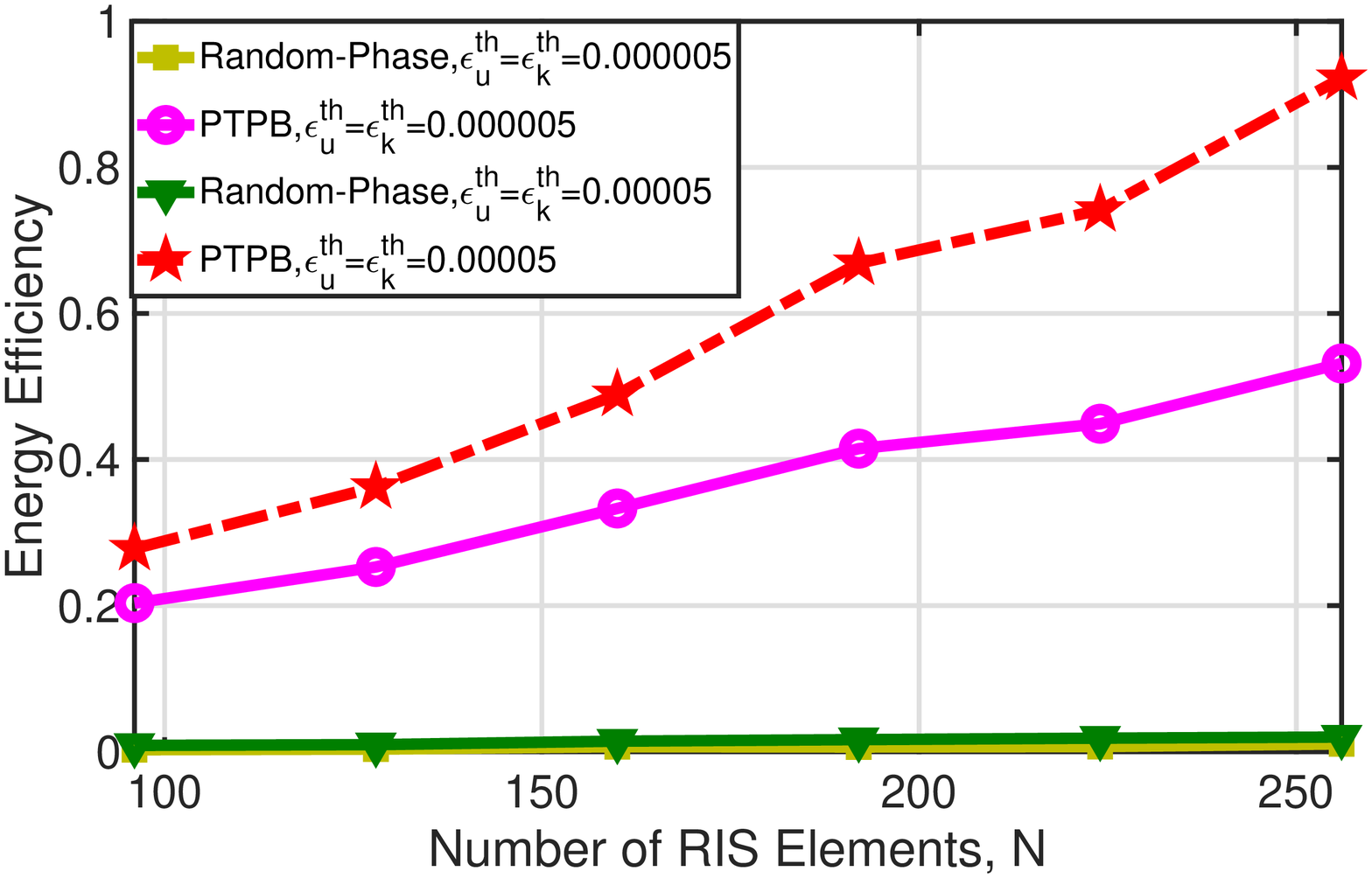}
			\caption{{Achieved energy efficiency versus the number of RIS elements $N$.}}
			\label{fig_ris_phaseshift}
		\end{minipage}
	\end{figure}

	
	Besides, we adopt the evaluation metric, energy efficiency computed by (\ref{eq:formualted_problem}a), to further verify the effectiveness of the designed RIS phase shift optimization strategy. We plot the obtained energy efficiency of the comparison algorithms under different numbers of RIS elements and reliability constraints in Fig. \ref{fig_ris_phaseshift}. By referring to Fig. \ref{fig_ris_channel_gain}, the achieved BS-HAP-UAV channel gains by both Random-Phase and Zero-Phase are extremely small. The results in Fig. \ref{fig_ris_phaseshift} illustrate that the proposed algorithm can achieve higher energy efficiency and satisfy the reliability requirement with a lower transmit power budget.  
 

	From the results in Figs. \ref{fig_ris_channel_gain} and \ref{fig_ris_phaseshift}, we can draw the conclusion that PTPB can effectively deal with the phase alignment issue in a RIS-integrated multipath EM propagation environment.
	
	\subsubsection{Results of Transmit Power and Blocklength Optimization}
	To verify the effectiveness of the proposed algorithm regarding joint transmit power and blocklength optimization, we compare it with two benchmark algorithms, i.e., MTP and MBL. To this end, we evaluate the impact of the number of the RIS elements $N$, {the side length of $\mathcal G$}, and the UAV transmit power budget $P_U$, respectively. 
	
	\begin{figure}[!t]
		\centering
		\subfigure[$\varepsilon^{\rm th}_u=\varepsilon^{\rm th}_k = 0.00005$]{\includegraphics[width=2.4in]{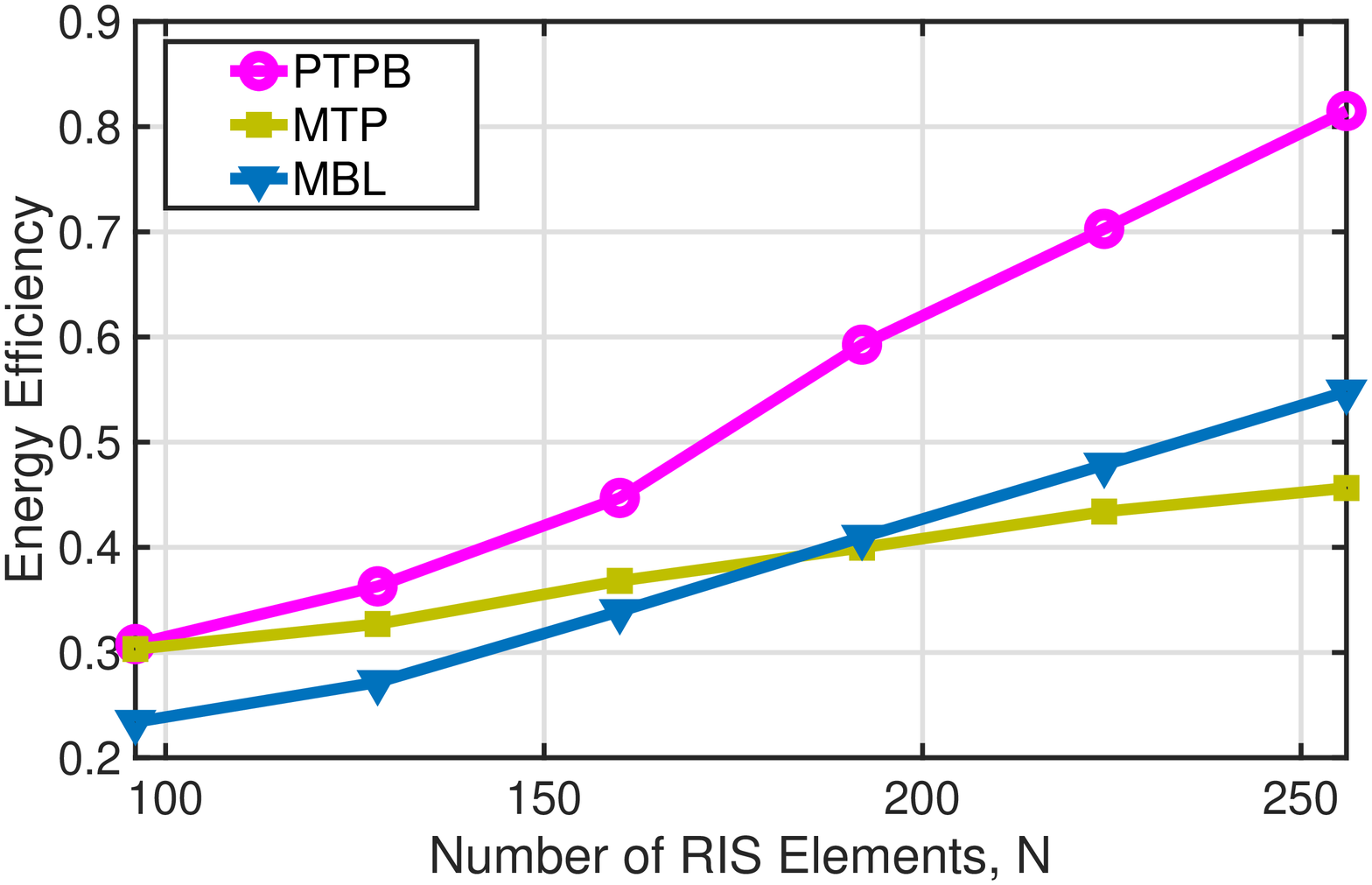}%
			\label{fig_ee_first_case_RIS}}
		\hfil
		\subfigure[$\varepsilon^{\rm th}_u=\varepsilon^{\rm th}_k = 0.000005$]{\includegraphics[width=2.4in]{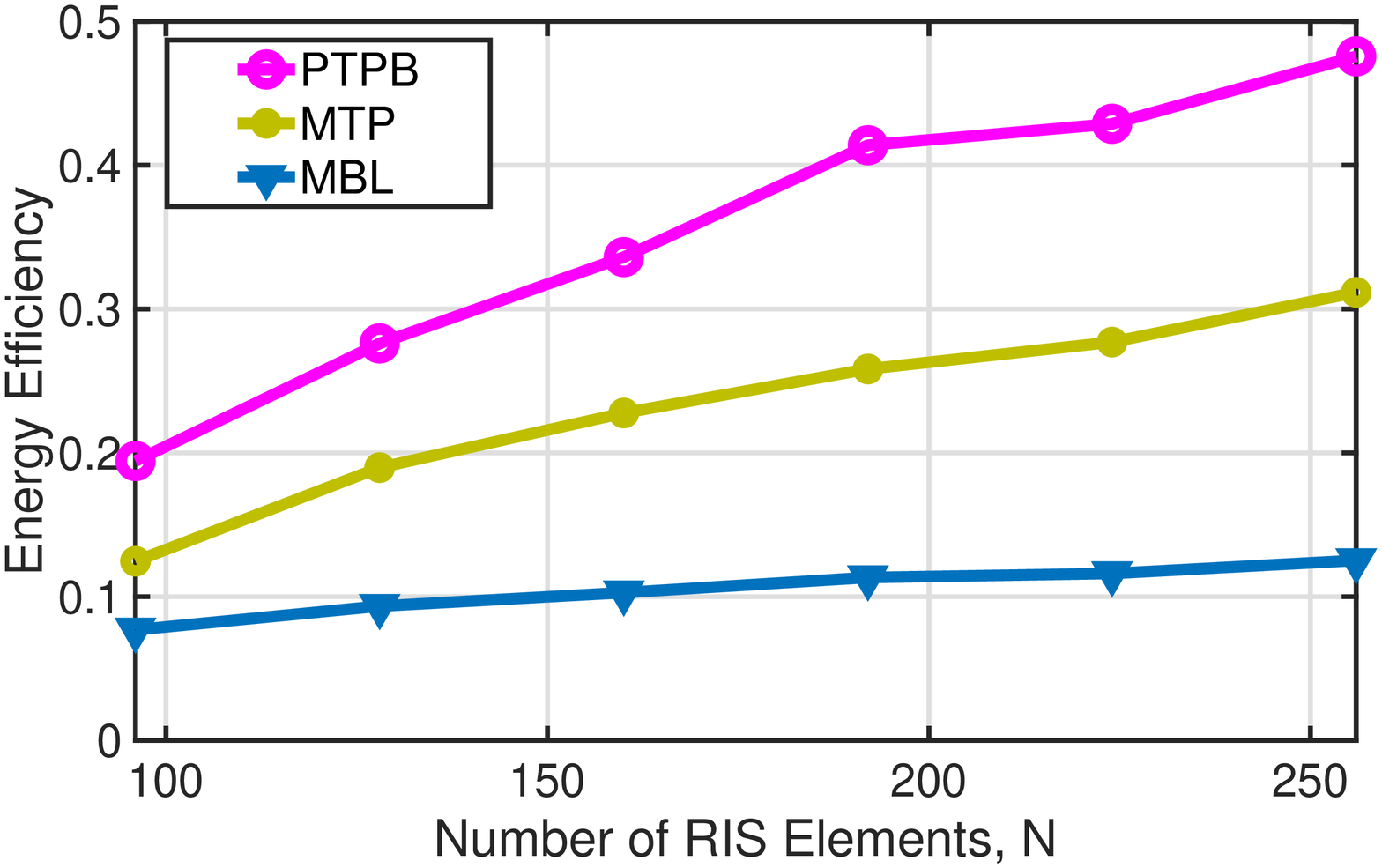}%
			\label{fig_ee_second_case_RIS}}
		\caption{Achieved energy efficiency versus the number of RIS elements $N$.}
		\label{fig_EE_versus_RIS_num}
	\end{figure}
	Fig. \ref{fig_EE_versus_RIS_num} shows the achieved energy efficiency of all comparison algorithms over the number of RIS elements. From this figure, we can obtain the following observations: 1) PTPB outperforms the other two benchmarks under diverse $N$, $\varepsilon^{\rm th}_u$, and $\varepsilon^{\rm th}_k$. The achieved energy efficiency by {all comparison algorithms} increases with an increasing number of RIS elements. This is due to the improvement in the {BS-HAP-UAV} channel quality, and then fewer resources (including transmit power and blocklength) are required to satisfy URLLC requirements. 
	2) When a more stringent reliability requirement, i.e., {$\varepsilon^{\rm th}_u=\varepsilon^{\rm th}_k = 0.000005$}, is enforced, the obtained energy efficiency of all comparison algorithms will decrease. Obviously, more resources will be consumed to satisfy the more stringent reliability requirement. 
	{3) Under a more stringent reliability constraint, the energy efficiency obtained by MBL does not increase effectively with improved channel conditions. In this case, the transmit power cannot be significantly reduced owing to the strictly limited blocklength.}

	\begin{figure}[!t]
		\centering
		\subfigure[$\varepsilon^{\rm th}_u=\varepsilon^{\rm th}_k = 0.00005$]{\includegraphics[width=2.4in]{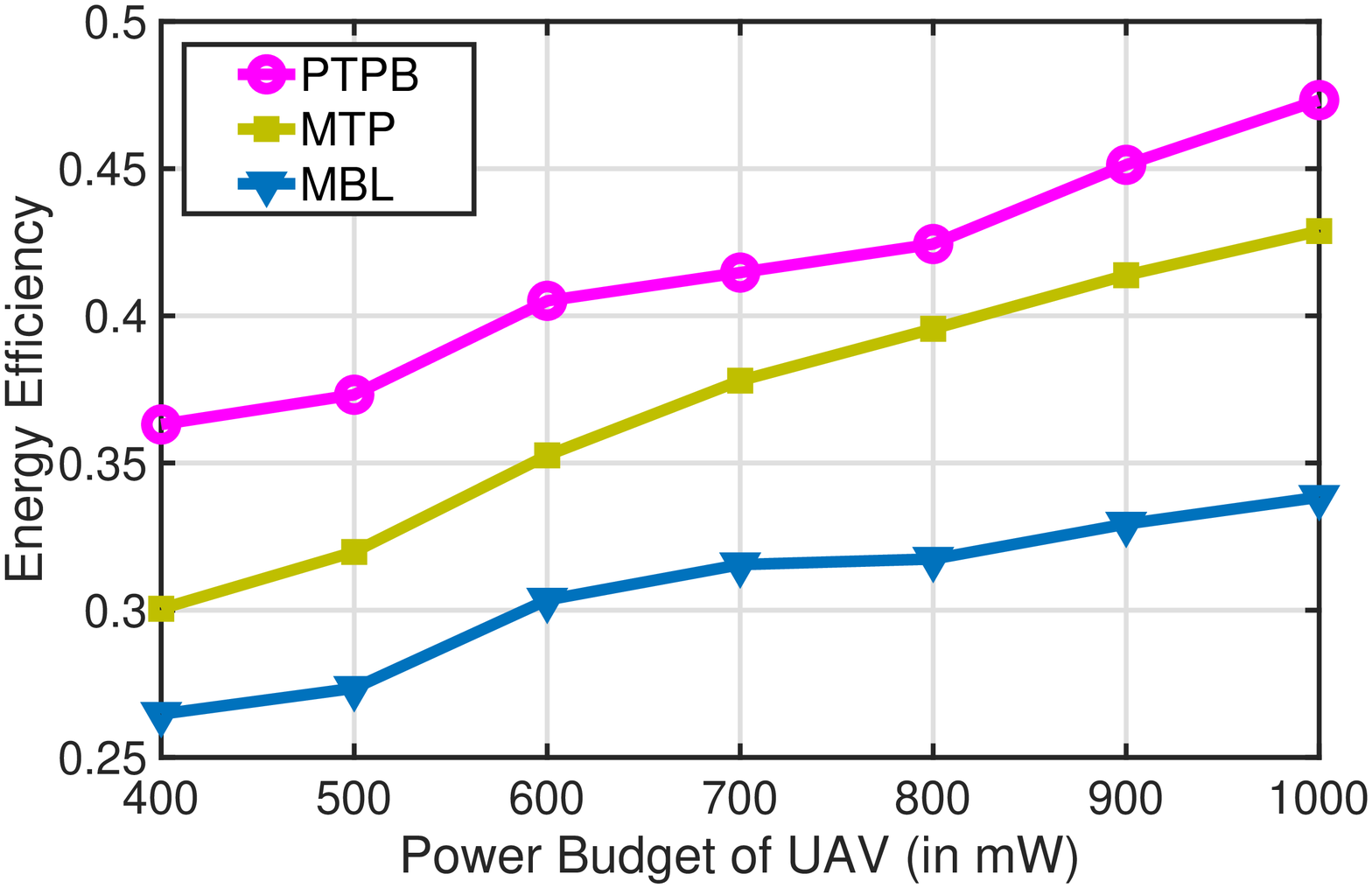}%
			\label{fig_ee_first_case_UE}}
		\hfil
		\subfigure[$\varepsilon^{\rm th}_u=\varepsilon^{\rm th}_k = 0.000005$]{\includegraphics[width=2.4in]{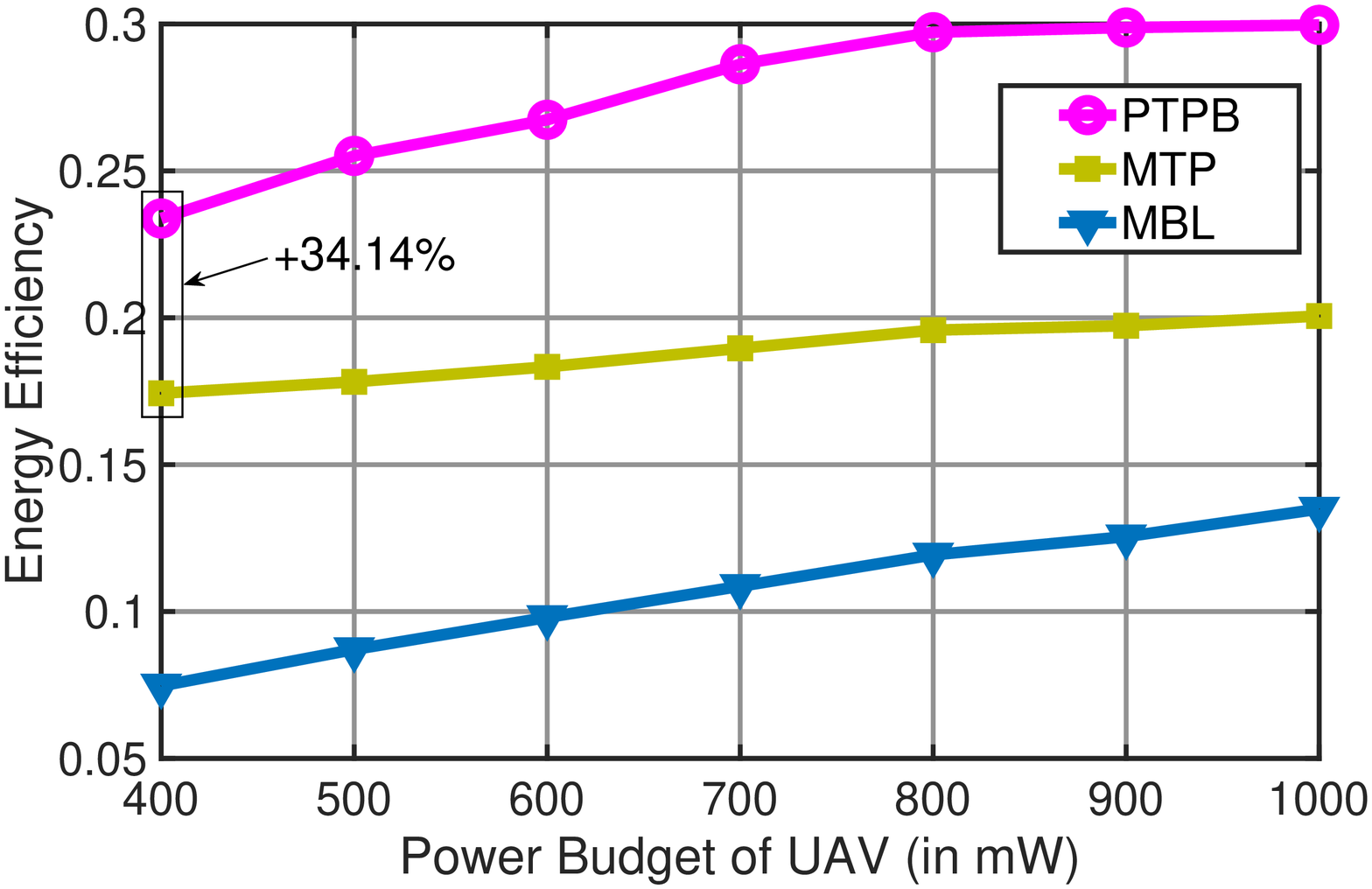}%
			\label{fig_ee_second_case_UE}}
		\caption{Achieved energy efficiency versus the transmit power budget of UAV.}
		\label{fig_EE_versus_UAV_num}
	\end{figure}
	Finally, we plot the tendency of the obtained energy efficiency of all comparison algorithms over the transmit power budget of UAV {$P_U \in \{400,500,600,700,800,900,1000\}$ mW} in Fig. \ref{fig_EE_versus_UAV_num}. In this simulation, we set the maximum blocklength {$ B_{u}^{max}=B^{max}=1000$ bits.} The following observations can be obtained from this figure: 1) PTPB achieves the highest energy efficiency. For instance, under a tighter reliability constraint, the system energy efficiency can be improved by at least {$34.14\%$} by executing the proposed PTPB algorithm. 
	{2) The performance of MBL degrades more compared to MTP when a more stringent reliability constraint is imposed. 
    The reasons are as follows: compared to the case of $\varepsilon^{\rm th}_u=\varepsilon^{\rm th}_k = 0.00005$, MBL needs to allocate higher transmit power to satisfy the tighter reliability requirement in the case of $\varepsilon^{\rm th}_u=\varepsilon^{\rm th}_k = 0.000005$. Owing to the limited feasible region of blocklength and the adoption of the maximum transmit power, the obtained energy efficiency by MTP will not be significantly reduced when a more stringent reliability requirement is imposed. }
	
	Summarily, from the above simulation results, we can conclude that joint RIS phase shift, transmit power, and blocklength optimization can significantly improve the system energy efficiency. 
	
	\section{Conclusion}
	This paper proposed to deploy a RIS-integrated NSIN to provide energy-efficient URLLC services for remote robots, which leveraged the advantages of RIS and power control. For this aim, we formulated the energy-efficient URLLC service provision problem as a resource optimization problem aiming at maximizing the effective throughput and minimizing the system power consumption, subject to URLLC and physical resource constraints. 
This problem was difficult to be mitigated due to the unknown channel model and intractable analysis. 
To handle these challenges, we developed a novel channel estimation approach, derived an analysis-friendly expression of DEP, and decomposed the problem by analyzing its monotonicity. 
Based on the above results, we proposed a joint RIS phase shift, transmit power, and blocklength optimization algorithm, the computational complexity of which was also analyzed. 
Simulation results verified that the developed channel estimation approach could accurately estimate HAP-UAV channels and the proposed resource optimization algorithm was more energy-efficient than benchmarks. 
	This paper investigated a stationary scenario, and extending it to a mobility scenario (e.g., the UAV is moving) will be an interesting topic worthy of study in the near future. In this case, channel tracking and RIS phase shift optimization approaches adapted to moving UAVs need to be designed.

\begin{appendix}
\subsection{Pilot Matrix Design}
 From \cite{DBLP:journals/spl/MaYP15a}, we know that a partial orthogonal measurement matrix outperforms an i.i.d. Gaussian measurement matrix when exploring an orthogonal approximate message passing (OAMP) approach. Motivated by this observation, a partial discrete Fourier transform (DFT) random permutation (pDFT-RP) measurement matrix is chosen. Besides, in practice, we don't know the off-grid offsets in advance. Nevertheless, if we deploy a great number of antenna elements, the measurement matrix ${\bm{F}}(\Delta\bm {\phi}) \approx {\bm{F}}(\pmb 0)$. Therefore, we set the corresponding off-grid offsets to be zero when generating the pilot matrix.
    We then have
    \begin{equation}\label{eq_}
		{\bm {F}}(\Delta \bm{\phi}) \approx {\bm {F}}({\bm 0}) = \bm {SDR}
	\end{equation}
	where ${\bm{S}}=\left\{0,1 \right\}^{P \times N}$ is a selection matrix generated by randomly selecting and reordering $P$ rows of an ${N \times N}$ identity matrix, ${\bm{D}}\in\mathbb{C}^{N \times N}$ is the DFT matrix, ${\bm{R}}\in\mathbb{C}^{N \times N}$ is a random permutation matrix generated by randomly reordering an ${N \times N}$ identity matrix. Note that the measurement matrix comprises the pilot matrix, BS-HAP channel matrix, and RIS phase shift matrix. We can compute the channel state information (CSI) of the BS-HAP channel in advance and can impose $\theta_n = -2\pi (n-1){\bar d}_{\rm RIS}( {\sin}({\arctan(\frac{||[x_H, y_H]^{\rm T}-[x_U, y_U]^{\rm T}||_2}{|h_H - h_U|})})-{\sin}(\upsilon_{{\rm AoA}}))$, $\forall n$, in the stage of designing the pilot matrix. 
	Then, the pilot matrix to be transmitted by the BS can be expressed as 
	\begin{equation}\label{eq_}
		{{\bm{U}_0^{\rm H}}= {\bm {SDR}}({{{\bm {\bar H}}^{\rm H}}{{\bm {\Theta}}^{\rm H}}\bm {A}(\bm{0})})^{-1}}
	\end{equation}
 
 \subsection{Proof of Theorem 1}
 Unfortunately, we can't obtain the closed-form expression of the log-likelihood function ${\rm ln}p({\bm{y}},\Delta\bm{\phi})$. As a result, the gradient descent method can't be directly utilized to gain the optimal solution $\Delta\bm{ \phi}^*$. To address this challenging issue, we utilize the in-exact majorization-minimization (MM) algorithm in \cite{DBLP:journals/tsp/DaiLL18} to approximate the optimal solution. As a general form of the expectation-maximization algorithm, the key idea of the in-exact MM algorithm is to iteratively generate a progressively refined (log-likelihood) function (a lower bound of ${\rm ln}p({\bm{y}},\Delta\bm{\phi})$) and approximate the optimal solution by continuously improving the lower bound. Although the expectation-maximization algorithm can't guarantee the theoretical optimality, it can guarantee to converge to a stationary point \cite{wu1983convergence}. 
	Specifically, let the surrogate function of ${\rm ln}p({\bm{y}},\Delta\bm{\phi})$ be $r(\Delta\bm{\phi};\Delta\bm {\dot\phi})$, which indicates that this progressively refined function is constructed at a given point $\Delta\bm{\dot\phi}$. As the lower bound, $r(\Delta\bm{\phi};\Delta\bm {\dot\phi})$ should satisfy the following three conditions:
	\begin{equation}\label{eq_}
		r(\Delta\bm{\phi};\Delta\bm {\dot\phi})\leq {\rm ln}p({\bm{y}},\Delta\bm{ \phi}),\forall \Delta\bm {\phi}
	\end{equation}
	\begin{equation}\label{eq_}
		r(\Delta\bm {\dot\phi};\Delta\bm {\dot\phi}) = {\rm ln}p({\bm y},\Delta\bm {\dot\phi})
	\end{equation}
	\begin{equation}\label{eq_}
		\displaystyle \frac{\partial r(\Delta\bm {\phi};\Delta\bm {\dot\phi})}{\partial \Delta\bm{ \phi}}|_{ \Delta\bm{ \phi}=\Delta\bm {\dot\phi}} = \frac{\partial {\rm ln}p(\bm{y},\Delta\bm{\phi})}{\partial  \Delta\bm{\phi}}|_{ \Delta\bm{\phi}=\Delta\bm {\dot\phi}}
	\end{equation}
	
	It's not difficult to verify that the following constructed progressively refined function satisfies all the above three conditions
	\begin{equation}\label{eq_surrogate_r}
		\displaystyle  r(\Delta\bm{\phi};\Delta\bm {\dot\phi})=\int p({\bm{x}}|{\bm{y}},\Delta\bm{\dot \phi}){\rm ln}\frac{p({\bm{x}},{\bm{ y}},\Delta\bm{\phi})}{p({\bm{x}}|{\bm{y}},\Delta\bm{\dot \phi})}d{{\bm{x}}}
	\end{equation}
	
	With the constructed function, we can iteratively update $\Delta\bm{\phi}$ by
	\begin{equation}\label{eq_arg_max_phi}
		\Delta\bm{\phi}^{i+1}= {\rm arg} \max \limits_{\Delta\bm{ \phi}}r(\Delta\bm{ \phi};\Delta\bm{\phi}^{i})
	\end{equation}
	where $i$ represents the $i$-th iteration and $\Delta\bm{\phi}^{i}$ stands for the corresponding value of $\Delta\bm{\phi}$.
	
	Nevertheless, it's quite difficult to obtain the optimal $\Delta\bm{\phi}^{i+1}$ as (\ref{eq_arg_max_phi}) is non-convex. To solve this problem effectively, we explore a gradient ascend method, and $\Delta\bm{\phi}^{i+1}$ can be updated by 
	\begin{equation}\label{eq_update_phi}
		\displaystyle \Delta\bm{ \phi}^{i+1}= \Delta\bm{\phi}^{i} + \alpha_i\frac{\partial r(\Delta\bm{\phi};\Delta\bm{\phi}^i)}{\partial \Delta\bm{\phi}}|_{\Delta\bm{\phi}= \Delta\bm{\phi}^i}
	\end{equation}
	where $\alpha_i$ is the update stepsize. Specifically, given the estimated $\hat p({\bm{x}})=\mathcal{CN}(0;{\bm{x}}^{post}_{B},v_B^{post}\bm{I})$ at the $i$-th iteration, the surrogate function of (\ref{eq_surrogate_r}) can be computed by
	\begin{equation}
		\begin{array}{l}
			\hat r(\Delta \bm{\phi} ;\Delta {\bm {\phi} ^i}) \propto {{\rm{E}}_{\hat p(\bm x|\bm y,\Delta \bm \phi )}}[ - \frac{1}{{\sigma _e^2}}||\bm{y} - \bm F(\Delta \bm {\phi} )\bm x|{|^2}]\\
			\propto  - \frac{1}{{\sigma _e^2}}(||\bm y - \bm F(\Delta \bm \phi )\bm x_B^{post}|{|^2} + v_B^{post}{\rm{tr}}(\bm F(\Delta \bm \phi )\bm F{(\Delta \bm \phi )^{\rm H}}))
		\end{array}
	\end{equation}
	
	The derivative of $\hat r(\Delta\bm{\phi};\Delta\bm {\phi}^{i})$ in (\ref{eq_update_phi}) can be calculated as $\bm{ \psi}^{(i)}_{\Delta\bm{ \phi}}=[\psi^{(i)}(\Delta \phi_1),...,\psi^{(i)}(\Delta \phi_N)]^{T}$, where	
	\begin{equation}
		\begin{array}{l}
			{\psi ^{(i)}}(\Delta {\phi _n}) = \\
			2{\rm{Re}}({{\bm{a}}^\prime }{({{\hat \omega }_n} + \Delta {\phi _n})^{\rm H}}{{\bm{U}}_{\bm{0}}}{{\bm{\bar H}}{\bm{\Theta}}}{({{\bm{U}}_{{0}}}{{\bm{\bar H}}{\bm{\Theta}}})^{\rm H}}{\bm{a}}({{\hat \omega }_n} + \Delta {\phi _n})) \\
			\times\varphi _1^{(i)} + 2{\rm{Re}}({{\bm{a}}^\prime }{({{\hat \omega }_{{n}}} + \Delta {\phi _{{n}}})^{\rm H}}{{\bm{U}}_{{0}}}{{\bm{\bar H}}{\bm{\Theta}}}\bm \varphi _2^{(i)})
		\end{array}  
	\end{equation}
	with  $\varphi_1^{(i)}=-\frac{1}{\sigma_e^2}(|{x}^{post}_{B,n}|^2+v_B^{post})$, ${\bm{ \varphi}}_2^{(i)}=\frac{1}{\sigma_e^2}({x}^{post}_{B,n})*{\bm{ y}}_{-n}$, ${\bm y}_{-n}={\bm{ y}}-({\bm{U_0}}{\bm{\bar H}}{\bm{\Theta}})^{\rm H}\sum_{j \neq n}a(\widehat\omega_j+\Delta\phi_{j}){x}^{post}_{B,j}, {\bm{ a}}'(\widehat\omega_n+\Delta\phi_{n})={\rm d}{\bm{ a}}(\widehat\omega_n+\Delta\phi_{n})/{\rm d}\Delta\phi_{n}$. 
	
	Then we can update the off-grid offsets as 
	\begin{equation}\label{eq_update_phix}
		\displaystyle \Delta\bm{ \phi}^{i+1}= \Delta\bm{\phi}^{i} + \alpha_i\bm{ \psi}^{(i)}_{\Delta\bm{ \phi}}
	\end{equation}
 
 To obtain high-precision estimation results, we utilize the backtracking linear search method, which can avoid oscillating results or slow convergence by adaptively adjusting the iteration stepsize. This completes the proof.

 \subsection{Proof of Lemma \ref{lem_beamformer}}
 Define the matrix $\bm{\Lambda}=\bm{ h}^{\rm H}\bm{\Theta}\bm{H}$. Then, we can rewrite ${ SNR}_u$ as
	\begin{equation}
		{SNR}_u = P_bG\bm{v}^{\rm H}\bm{\Lambda}^{\rm H}\bm{\Lambda}\bm{ v}/\sigma_0^2
	\end{equation}

	According to (\ref{eq_H_matrix}) and (\ref{eq_h_multipath}), $\bm{\Lambda}^{\rm H}\bm{\Lambda}$ is a rank-1 martix with the following structure
	\begin{equation}
		\bm{\Lambda}^{\rm H}\bm{\Lambda}=D\bm{a}_{\rm BS}(\upsilon_{{\rm AoD}})\bm{ a}_{\rm BS}(\upsilon_{{\rm AoD}})^{\rm H}
	\end{equation}
	where $D$ is a positive constant. 
	
	The non-zero eigenpair of $\bm{\Lambda}^{\rm H}\bm{\Lambda}$ can then be expressed as 
	\begin{equation}
		\displaystyle  (E_{val}, E_{vec})=(D||\bm{a}_{\rm BS}(\upsilon_{{\rm AoD}})||^2_2, \frac{\bm{a}_{\rm BS}(\upsilon_{{\rm AoD}})}{||\bm{a}_{\rm BS}(\upsilon_{{\rm AoD}})||_2})
	\end{equation}
	
	Owing to the unit magnitude of $\pmb{ v}$ and the Rayleigh-Ritz theorem, $\pmb{ v}$ can be calculated by 
	\begin{equation}
		\displaystyle \bm{ v}=\frac{\bm{a}_{\rm BS}(\upsilon_{{\rm AoD}})}{||\bm{a}_{\rm BS}(\upsilon_{{\rm AoD}})||_2}
	\end{equation}
 
This completes the proof.

 \subsection{Proof of Lemma \ref{cal_pb_bu}}
{To tackle this issue,} we propose to decompose (\ref{eq:BS_layer_opt}) into two subproblems, named \emph{BS transmit power control} subproblem and \emph{BS blocklength optimization} subproblem. 
	
	\textbf{1-a) BS transmit power control:}
	Given $b_u$, ${P}_u$ and $\{b_k\}$, the BS power control subproblem can be reduced to
	\begin{subequations}\label{eq:origin_BS_power}
		\begin{alignat}{2}
			&  \max \limits_{P_b}  \text{ } \min \limits_{k \in {\mathcal K}} \text{ } \displaystyle   \frac{R_u(1-\epsilon_u)+R_k(1-\epsilon_k)}{{\bar p}_b+{\bar p}_u}  \\
			& {{\rm s.t:}} { \text{ } 0 < P_b\leq P_{B},\text{ } \Omega(SNR_u) \le \varepsilon_u^{\rm th}}
		\end{alignat}
	\end{subequations}
	
	Recall that $\Omega({SNR}_u)$ is a piecewise function. Thus, we need to discuss the power control subproblem in three different cases closely related to the value of ${SNR}_u$. We can directly exclude the case of ${SNR}_u\leq {\gamma_u}-\frac{1}{{\chi_u}}$ as it will result in an unacceptable DEP. 
	Then, we discuss the remaining two cases, respectively. 
	
	\textbf{Case I:} ($SNR_u \ge SNR_{\rm up}$) In this case, we have $\varepsilon_u = 0$, and the BS power control subproblem can be reformulated as 
	\begin{subequations}\label{eq:case1_power_control}
		\begin{alignat}{2}
			&  \max \limits_{P_b} \text{ }  \min  \limits_{k \in {\mathcal K}} \text{ } \displaystyle   \frac{R_u+R_k(1-\epsilon_k)}{{\bar p}_b+{\bar p}_u}  \\
			& {\rm s.t:} \text{ } \displaystyle  {\bar p}_b\geq  \frac{SNR_{\rm up} }{P_B\Delta B}, {\text{ } 0 < P_b\leq P_{B}.}
		\end{alignat}
	\end{subequations}
	where $\Delta B=G|\bm{ h}^{\rm H}\bm{\Theta} \bm{Hv}|^2/\sigma_0^2$. As the numerator of (\ref{eq:case1_power_control}a) is not related to ${\bar p}_u$, the objective function of (\ref{eq:case1_power_control}) monotonically decreases with ${P}_b$. Thus, given $b_k$, the optimized BS transmit power (denoted by $P_b(b_k)$) can be computed by 
	\begin{equation}\label{eq_optimal_BS_power}
		\displaystyle  P_b(b_u)={\rm \min} \left\{P_B, \frac{SNR_{\rm up}}{\Delta B} \right\}
	\end{equation}
	
	\textbf{Case II:} ($SNR_{\rm low} < SNR_{u} < SNR_{\rm up}$) In this case, We can reformulate (\ref{eq:origin_BS_power}) as
	\begin{subequations}\label{eq:case2_BS_power}
		\begin{alignat}{2}
			&  \max \limits_{P_b} \text{ }  \min  \limits_{k \in {\mathcal K}} \text{ } \displaystyle   \frac{R_u(1-\epsilon_u)+R_k(1-\epsilon_k)}{{\bar p}_b+{\bar p}_u}  \\
			& {\rm s.t:} \text{ }\displaystyle  \frac{SNR_{\rm low}}{P_B\Delta B} \leq {\bar p}_b\leq  \frac{SNR_{\rm up} }{P_B\Delta B}, {\text{ } 0 < P_b\leq P_{B}.}
		\end{alignat}
	\end{subequations}
	
	Generally, the monotonicity of (\ref{eq:case2_BS_power}a) is hard to be determined. Fortunately, we observe that the feasible region ${\mathcal P}_b$ of (\ref{eq:case2_BS_power}) is rather narrow, in which the numerator of (\ref{eq:case2_BS_power}a) varies greatly, while the denominator remains almost unchanged. Therefore, the monotonicity of (\ref{eq:case2_BS_power}a) is dominated by its numerator. The numerator is then a monotonically increasing function in ${\mathcal P}_b$, and so is (\ref{eq:case2_BS_power}a). It further indicates that the two discussed cases have the same solution that is given in (\ref{eq_optimal_BS_power}). 
	
	\textbf{1-b) BS blocklength optimization problem:}
	Given $P_b(b_u)$, $P_u$, and $\{b_k\}$, we can formulate the BS blocklength optimization problem as
{\begin{subequations}\label{eq:BS_blocklength_opt}
		\begin{alignat}{2}
			&  \max \limits_{b_u}  \text{ } \min \limits_{k \in {\mathcal K}} \displaystyle \frac{R_u(1-\epsilon_u)+R_k(1-\epsilon_k)}{P_b(b_B)/P_B+\overline{p}_u}  \\
			& {\rm s.t:} {\text{ } B_{u}^{\min} \leq b_u \leq B_{u}^{\max}, \text{ } b_u \in \mathbb{Z^+},\text{ } \Omega(SNR_u) \le \varepsilon_u^{\rm th}.}
		\end{alignat}
\end{subequations}}
	
	Note that the obtained $P_b(b_u)$ is a function of $b_u$. After substituting $P_b(b_u)$ into (\ref{eq:BS_blocklength_opt}), we can obtain an objective function regarding variable $b_u$. However, the obtained objective function is non-concave even given a relaxed variable $b_u$. Fortunately, considering that the feasible region of (\ref{eq:BS_blocklength_opt}) is small, it will be a good choice to achieve its optimal solution by exhaustively searching.
	
	\textbf{Remark 2:} Observe that the optimized $P_b(b_u)$ is a function $b_u$. After obtaining $b_u$ using an exhaustive search method, we can directly obtain the optimized $P_b(b_u)$ by (\ref{eq_optimal_BS_power}). It indicates that we can solve the \emph{BS-Layer Optimization} problem in just two iterations. 
This completes the proof.

\subsection{Proof of Lemma \ref{lem_UAV_transmit_power}}
The objective function of (\ref{eq:UAV_power_control}) is a complex fraction, which makes it hard to be directly optimized. To address this issue, we define the optimal $\eta_{\rm EE}^{\star}$ as
	\begin{equation}\label{eq_}
		\displaystyle \eta_{\rm EE}^{\star}=\min  \limits_{k \in {\mathcal K}} \text{ }  \displaystyle \frac{R_u(1-\epsilon_u)+R_k({\bar p}_u^{\star})(1-{\epsilon}_k({\bar p}_u^{\star}))}{{\bar p}_b+{\bar p}_u^{\star}}
	\end{equation}
	where ${\bar p}_u^{\star} = P_u^{\star}/{P_U}$ is the optimal solution of (\ref{eq:UAV_power_control}). Define ${\mathcal P}$ as the feasible region of (\ref{eq:UAV_power_control}), and $\eta_{\rm EE} = \min  \limits_{k \in {\mathcal K}} \text{ } \displaystyle  \frac{R_u(1-\epsilon_u)+R_k(1-\epsilon_k)}{{\bar p}_b+{\bar p}_u} = \min  \limits_{k \in {\mathcal K}} \text{ }  \frac{V_k(\bar p_u)}{S(\bar p_u)}$. Denote by $\eta_{\rm EE}^{\star}$ and $(\bar p_u^{\star} P_U) \in {\mathcal P}$ the optimal energy efficiency and the optimal UAV transmit power, respectively. Then, we have
	\begin{equation}\label{eq_}
		\displaystyle \eta_{\rm EE}^{\star}=  \displaystyle \frac{\min  \limits_{k \in {\mathcal K}}  V_k(\bar p_u^{\star})}{S(\bar p_u^{\star})} \ge  \displaystyle \frac{\min  \limits_{k \in {\mathcal K}}  V_k(\bar p_u)}{S(\bar p_u)}
	\end{equation}
	
	Accordingly, the following equation holds
	\begin{equation}\label{eq_}
		\displaystyle 
		\max  \limits_{P_u} \text{ } {\min  \limits_{k \in {\mathcal K}} \text{ } V_k(\bar p_u)} - \eta_{\rm EE}^{\star}{S(\bar p_u)} = 0
	\end{equation}
	
	On the contrary, suppose that $(\bar p_u^{\star}P_U)$ is the optimal UAV transmit power of the following optimization problem
	\begin{subequations}\label{eq:equivalent_UAV_power_control}
		\begin{alignat}{2}
			&  \max \limits_{P_u^\prime } \text{ } \min  \limits_{k \in {\mathcal K}} \text{ }   \displaystyle \text{ }  V_k(P_u^\prime/{P_U})- \eta_{\rm EE}^{\star}S(P_u^\prime/{P_U})  \\
			& {\rm s.t:}  {\text{ } 0 < P_u \leq P_{U}}\\ 
			& {\rm Constraint \text{ } (\ref{eq:formualted_problem}c).}
		\end{alignat}
	\end{subequations}
	and the following equation holds
	\begin{equation}\label{eq_}
		\displaystyle {\min  \limits_{k \in {\mathcal K}} \text{ } V_k(\bar p_u^{\star})} - \eta_{\rm EE}^{\star}{S(\bar p_u^{\star})} = 0
	\end{equation}
	
	As (\ref{eq:UAV_power_control}) and (\ref{eq:equivalent_UAV_power_control}) have the same feasible region ${\mathcal P}$, for any UAV transmit power $P_u^\prime \in {\mathcal P}$, the following inequality holds
	\begin{equation}\label{eq_contrary_proof}
		\begin{array}{l}
			\mathop {\min }\limits_{k \in {\mathcal K}} \text{ } {\rm{ }}{V_k}({P_u^\prime}/{P_U}) - \eta _{{\rm{EE}}}^ \star S({P_u^\prime}/{P_U})\\
			\le \mathop {\min }\limits_{k \in {\mathcal K}} \text{ } {\rm{ }}{V_k}(\bar p_u^ \star ) - \eta _{{\rm{EE}}}^ \star S(\bar p_u^ \star )\\
			= 0
		\end{array}
	\end{equation}
	
	From (\ref{eq_contrary_proof}), we obtain $\eta _{{\rm{EE}}}^ \star  = \frac{{\mathop {\min }\limits_{k \in {\mathcal K}}  {\rm{ }}{V_k}(\bar p_u^ \star )}}{{S(\bar p_u^ \star )}} \ge \frac{{\mathop {\min }\limits_{k \in {\mathcal K}}  {\rm{ }}{V_k}(P_u^{\prime}/{P_U})}}{{S({P_u^{\prime}}/{P_U})}}$. Then we can conclude that (\ref{eq:UAV_power_control}) and (\ref{eq:equivalent_UAV_power_control}) have the same optimal UAV transmit power.  {\rm $\eta_{\rm EE}^{\star}$ can be achieved if and only if} 
		\begin{equation}\label{eq_equivalent_expression}
			\begin{array}{l}
				\mathop {\max }\limits_{{P_u}} \mathop {\min } \limits_{k \in {\mathcal K}} \text{ } R_u(1-\epsilon_u) + {R_k}(1 - \epsilon{_k}) - {\eta_{\rm EE}^{\star} }\left( {{{\bar p}_b} + {{\bar p}_u}} \right)\\
				= \mathop {\min }\limits_{k \in {\mathcal K}} \text{ }  R_u(1-\epsilon_u) + {R_k}({{\bar p}_u^{\star}})(1 - \epsilon{_k}({{\bar p}_u^{\star}})) - {\eta_{\rm EE} ^{\star}}\left( {{{\bar p}_b} + {{\bar p}_u^{\star}}} \right) \\
				= 0
			\end{array}
		\end{equation}
	
	{Following the conclusion in (\ref{eq_equivalent_expression}),} we can transform the expression of the objective function of (\ref{eq:UAV_power_control}). Specifically, (\ref{eq:UAV_power_control}) can be equivalently transformed into
	\begin{subequations}\label{eq:reformu_UAV_power_control}
		\begin{alignat}{2}
			&  \max \limits_{P_u}  \min  \limits_{k \in {\mathcal K}} \text{ } \displaystyle  R_u(1-\epsilon_u) + R_k(1-{\epsilon}_k)-\eta_{\rm EE}({{\bar p}_b+{\bar p}_u})  \\
			& {\rm s.t:} \text{ } \displaystyle {\bar p}_u \geq \frac{{2\gamma_k} \sigma_{k}^2}{|g_k^L|^2 \varepsilon^{\rm th}_kP_U},\text{ }\forall k\in {\mathcal K},\text{ }{0 < P_u \leq P_{U}.}
		\end{alignat}
	\end{subequations}
	
	Note that (\ref{eq:reformu_UAV_power_control}) is a typical max-min optimization problem. Then, we introduce an auxiliary variable $y$ and equivalently transform (\ref{eq:reformu_UAV_power_control}) into
\begin{subequations}\label{eq:equi_reform_UAV_power_control}
		\begin{alignat}{2}
			&  \max \limits_{P_u,y} \text{ } y  \\
			& {\rm subject \text{ } to:} \nonumber \\
			& \displaystyle \frac{{2\gamma_k} \sigma_{k}^2}{|g_k^L|^2 {\bar p}_uP_U} \leq 1+\frac{R_u(1-\epsilon_u)-y-\eta_{\rm EE}({{\bar p}_b+{\bar p}_u})}{R_k},\forall k\in {\mathcal K} \\
			& {\rm Constraint \text{ }  (\ref{eq:reformu_UAV_power_control}b).}
		\end{alignat}
	\end{subequations}
	
	It can be confirmed that (\ref{eq:equi_reform_UAV_power_control}) is a rotated quadratic cone programming problem that can be efficiently optimized by some commercial optimization tools such as MOSEK. 
	Therefore, we can conclude the steps of optimizing (\ref{eq:UAV_power_control}) in Algorithm \ref{alg_UAV_power_control}. 
This completes the proof.

\subsection{Discussion on Algorithm Complexity}
Here, we discuss the computational complexity of Algorithm \ref{alg_PTPB} in detail, which consists of three contributors: 
1) \emph{Channel estimation}: R-OAMP includes a module A, a module B, and a parameter updating module. Module A is an LMMSE estimator with a complexity of $O(PN)$.
Module B is an MMSE estimator where a sum-product approach is leveraged, which thus has a low computational complexity of $O(N)$. 
The complexity of updating offsets is $O(PN^2)$.
Considering that the above modules should be iteratively executed for no more than $I_c^{\max}$ times. The computational complexity of R-OAMP is $O(f_1) = O(I_c^{\max}(PN+N+PN^2))$ in the worst-case. However, empirical evidence shows that R-OAMP usually
converges in 30 iterations.
2) \emph{RIS phase shift optimization}: 
The complexity of solving (\ref{eq:optimize_RIS}) using an interior-point method dominates the complexity of this contributor and is $O((1+L)^{3.5})$ \cite{DBLP:books/daglib/0094154}.
3) \emph{Joint power and blocklength optimization}: 
This joint optimization model includes the BS-Layer optimization and the UAV-Layer optimization. As the iterative method in BS-Layer optimization will converge in two iterations, the complexity of BS-Layer optimization is dominated by the exhaustive search approach and is $O(f_2) = O((B_u^{\max} - B_u^{\min})S_1 + \log_2(B_u^{\max} - B_u^{\min}))$, where $S_1$ is the complexity of computing (\ref{eq:BS_blocklength_opt}a). 
For UAV-Layer optimization, it needs to iteratively optimize (\ref{eq:UAV_power_control}) and (\ref{eq:UAV_block_opt}) for no more than $c_{\max}$ iterations. The complexities of optimizing UAV transmit power and blocklength are $O(r_{\max}2^{3.5})$ and $O((B^{\max} - B^{\min})S_2 + \log_2(B^{\max} - B^{\min}))$ per iteration, where $S_2$ is the complexity of computing (\ref{eq:UAV_block_opt}a). Thus, the complexity of UAV-Layer optimization is $O(f_3) = O(c_{\max}(r_{\max}2^{3.5} + (B^{\max} - B^{\min})S_2 + \log_2(B^{\max} - B^{\min}) ))$ in the worst-case.
Then, the total computational complexity of is $O(f_1 + I_{\max}f_2 + (1+L)^{3.5} + f_3)$ in the worst-case.

\end{appendix}

	\bibliographystyle{IEEEtran}
	\bibliography{hap_urllc}

\end{document}